\documentclass[sigplan,10pt,screen]{acmart}

\copyrightyear{2023} 
\acmYear{2023} 
\setcopyright{acmlicensed}\acmConference[SOSP '23]{ACM SIGOPS 29th Symposium on Operating Systems Principles}{October 23--26, 2023}{Koblenz, Germany}
\acmBooktitle{ACM SIGOPS 29th Symposium on Operating Systems Principles (SOSP '23), October 23--26, 2023, Koblenz, Germany}
\acmPrice{15.00}
\acmDOI{10.1145/3600006.3613152}
\acmISBN{979-8-4007-0229-7/23/10}

\usepackage{optidef} 
\usepackage{xspace}

\usepackage[T1]{fontenc}
\usepackage{subcaption}
\usepackage{xcolor}
\usepackage{colortbl}
\usepackage{pifont}
\usepackage[inline]{enumitem}
\usepackage{multirow}
\usepackage[textfont={it}, font={small,md,rm}]{caption}

\usepackage{tikz}
\newcommand*\circled[1]{\tikz[baseline=(char.base)]{
		\node[shape=circle,draw,inner sep=0pt] (char) {#1};}}

\usepackage{hyperref}
\hypersetup{
	colorlinks=true,      
	linkcolor=blue,       
	citecolor=magenta,    
	filecolor=cyan,       
	urlcolor=red          
}

\newcolumntype{L}[1]{>{\raggedright\arraybackslash}p{#1}}
\newcolumntype{C}[1]{>{\centering\arraybackslash}p{#1}}
\newcolumntype{R}[1]{>{\raggedleft\arraybackslash}p{#1}}

\DeclareMathOperator{\dplus}{+\kern -0.4em+}

\newenvironment{denseitemize}{
	\begin{itemize}[topsep=2pt, partopsep=0pt, leftmargin=1.5em]
		\setlength{\itemsep}{2pt}
		\setlength{\parskip}{0pt}
		\setlength{\parsep}{0pt}
	}{\end{itemize}}

\newenvironment{denseenum}{
	\begin{enumerate}[topsep=2pt, partopsep=0pt, leftmargin=1.5em]
		\setlength{\itemsep}{2pt}
		\setlength{\parskip}{0pt}
		\setlength{\parsep}{0pt}
	}{\end{enumerate}}

\def\ie{{i.e.}}
\def\eg{{e.g.}}

\newcommand{\name}{Oobleck\xspace}
\newcommand{\parabf}[1]{\medskip\noindent\textbf{#1}}

\begin{document}

\title{\name: Resilient Distributed Training of Large Models Using Pipeline Templates}

\author{Insu Jang}
\email{insujang@umich.edu}
\orcid{0009-0007-5206-2333}
\affiliation{%
	\institution{University of Michigan}
	\country{}
}
\author{Zhenning Yang}
\email{znyang@umich.edu}
\orcid{0009-0003-0813-5911}
\affiliation{%
	\institution{University of Michigan}
	\country{}
}
\author{Zhen Zhang}
\email{zhzhn@amazon.com}
\orcid{0000-0002-0164-0849}
\affiliation{%
	\institution{Amazon Web Services}
	\country{}
}
\author{Xin Jin}
\email{xinjinpku@pku.edu.cn}
\orcid{0000-0001-8741-5847}
\affiliation{%
	\institution{Peking University}
	\country{}
}
\author{Mosharaf Chowdhury}
\email{mosharaf@umich.edu}
\orcid{0000-0003-0884-6740}
\affiliation{%
	\institution{University of Michigan}
	\country{}
}



\begin{abstract}

\name enables resilient distributed training of large DNN models with guaranteed fault tolerance.
It takes a planning-execution co-design approach, where it first generates a set of heterogeneous \emph{pipeline templates} and instantiates at least $f+1$ logically equivalent pipeline replicas to tolerate any $f$ simultaneous failures.
During execution, it relies on already-replicated model states across the replicas to provide fast recovery.
\name provably guarantees that some combination of the initially created pipeline templates can be used to cover all available resources after $f$ or fewer simultaneous failures, thereby avoiding resource idling at all times. 
Evaluation on large DNN models with billions of parameters shows that \name provides consistently high throughput, and it outperforms state-of-the-art fault tolerance solutions like Bamboo and Varuna by up to $29.6\times$.
\end{abstract}

\maketitle

\section{Introduction}
\label{sec:introduction}

DNN models continue to become larger \cite{2022modelsize}.
Many recent advances in deep learning have been attributed to significant increases in model size to hundreds of billions of parameters and training on ever-growing datasets \cite{2020gpt3, 2023gpt4, 2022turingnlg, 2019dlrm}.
Recent studies suggest that a trillion-parameter model would require at least 2TB of memory simply to store model parameters, and tens or hundreds of TB for training~\cite{2020zero, 2021zerooffload, 2021zeroinfinity, 2022persia, 2022whale}. 
Naturally, scaling large model training has received intense attention over the past few years \cite{2019pipedream, 2021dapple, 2022turingnlg, 2022opt, 2022bloom}.
Distributed \emph{hybrid-parallel training} \cite{2021megatron,2022alpa} that combines model and data parallelism has emerged as the primary approach to training such large models.

Unfortunately, the likelihood of experiencing failures also increases with the scale and duration of training \cite{2019philly, 2017failuresinlargescale, 2022mlaaswild}.
The effect is further amplified by the synchronous nature of DNN training, which causes all participating devices to idle until the failed one has recovered, causing massive underutilization.
Indeed, teams from Meta, HuggingFace, and LAION report significant underutilization from failures when training large models \cite{2022opt, 2022bloom, 2022laion}.
Failure rates are even higher for training jobs that use spot instances in the cloud \cite{2023bamboo, 2022varuna}.

Existing frameworks have little systematic support for fault tolerance during hybrid-parallel training.
Ensuring continuous operation in the presence of failures fundamentally requires redundancy in one form or another. 
Model state redundancy in data-parallel training is the only form of ``free'' redundancy, because each worker already has a copy of the model states.
Most solutions harness the inherent redundancy provided by data parallelism and utilize its embarrassingly parallel nature to elastically change the number of GPUs while dynamically changing the global batch size \cite{2021coddl, torchelastic,2022aryl,2023ecrec}. 
However, they are unable to extend these benefits to hybrid parallelism and are limited only to data-parallel training.

In contrast, fault tolerance approaches tailored toward hybrid parallelism struggle to leverage any inherent redundancy.
Instead, they introduce additional redundancy in various forms; \eg, having a pool of standby GPUs \cite{2022bloom}, using checkpoints to reconfigure and restart \cite{2022varuna}, and performing redundant computations in anticipation of a possible failure \cite{2023bamboo}.
Essentially, they consider overhead during training vs. overhead to recover from failure(s), and choose one of the two extremes (\S\ref{sec:backgroundfaulttolerance}): if failures are infrequent, amortized overhead for reconfiguration would also be low \cite{2022varuna}; if failures are more frequent, then incurring some overhead during training may be more preferable than spending significant time in recovery \cite{2023bamboo}.

In this paper, we present \name, a fault-tolerant hybrid-parallel training framework.
It provides high training throughput, guaranteed fault tolerance, and fast recovery \emph{without} introducing additional overhead.
%
\emph{Pipeline templates} are at the core of \name's design.
A pipeline template is a specification of pipeline execution for a given number of nodes.
They are designed during the planning phase by \name's template generator and reused during execution by \name's execution engine.
All pipeline templates are logically equivalent yet physically heterogeneous; each has a different number of nodes and associated configurations that can be used to instantiate a pipeline for a given model.
\name uses one or more pipeline templates to create pipeline replicas to exploit the inherent model states redundancy across the replicas.
Pipelines affected by failures can reconstruct model states by copying missing layers from other replicas without having to restart from a checkpoint.

More specifically, given a training job starting with the number of maximum simultaneous failures to tolerate $f$,
\name's execution engine instantiates at least $f+1$ heterogeneous pipelines from the generated templates.
The fixed global batch is distributed proportionally to the computing capability of heterogeneous pipelines such that all pipeline replicas train roughly at the same rate.
Upon failures, \name avoids demanding analysis of finding a new optimal configuration by simply reinstantiating pipelines from the precomputed pipeline templates while achieving maximum node utilization.
This is always possible for $f$ or fewer failures because \name provably guarantees that a combination of pipelines generated from those precomputed templates can fully utilize all the remaining nodes.


We have implemented \name on top of PyTorch and HuggingFace Transformers~\cite{2022hftransformer} using components from DeepSpeed~\cite{2020deepspeed} and Merak~\cite{2023merak}.
We evaluate \name and compare its performance against Bamboo and Varuna across large models like GPT-3 with billions of parameters.
\name outperforms the state-of-the-art solutions by up to $29.6 \times$ as we consider different frequencies of failures, spot instance traces, and models of different sizes and computation complexity.

Overall, we make the following contributions in this paper.

\begin{denseitemize}
    \item We present \name, a novel framework for resilient distributed training that provides guaranteed fault tolerance and maximizes throughput.
    
    \item \name introduces pipeline templates to (re)instantiate pipelines.
        Pipeline templates allow quick failure recovery and utilization of all available GPUs.

    \item We implement and evaluate \name with several large models, \eg, variants of GPT-3, to demonstrate large improvements in terms of throughput and failure recovery.
\end{denseitemize}

\name is open-source and available on GitHub.\footnote{\url{https://github.com/SymbioticLab/Oobleck}}

\section{Background and Motivation}
\label{sec:background}

In this section, we briefly introduce hybrid parallelism that is commonly used for large model training.
We also discuss existing fault tolerance strategies for hybrid-parallel training and highlight their limitations.

\subsection{Hybrid Parallelism}
\label{sec:hybridparallelism}

As DNN models continue to grow in size and are trained on increasingly larger datasets \cite{2022opt, 2022bloom, 2022turingnlg}, using just \emph{data parallelism} or \emph{model parallelism} is often not enough to efficiently train a DNN.
Data parallelism splits and distributes input to multiple GPUs, but it requires each GPU to hold the entire model \cite{2021coddl}.
Model parallelism accommodates large model training by splitting the model across multiple GPUs --  for example, \emph{pipeline parallelism} splits the model into groups of layers called stages, and \emph{tensor parallelism} slices each model layer into several tensor chunks.
However, the former has pipeline bubble overheads that decrease compute utilization as the pipeline grows deeper~\cite{2022turingnlg}.
The latter suffers from high communication cost that cannot be hidden in computation, because it requires several all-reduce operations in the critical path of both forward pass and back-propagation~\cite{2020zero}.
Consequently, a combination of data and model parallelism techniques -- aka \emph{hybrid parallelism} -- is used in practice to train large DNN models \cite{2022turingnlg, 2022bloom, 2021dapple, 2021megatron, 2019pipedream}.

\subsection{Fault Tolerance in Distributed Training}
\label{sec:backgroundfaulttolerance}

Failures are the norm in distributed systems, and distributed DNN training is no exception. 
The probability of experiencing one or more failures increases with the increasing number of GPUs and the duration of training.
For instance, a Meta AI team suffered approximately 100+ hardware failures and had to do 100+ major restarts during OPT-175B training~\cite{2022opt}.
Because of the synchronous nature of distributed training, the cost of even one failure is \emph{multiplied}: all the GPUs must idle until the impact of failure has been mitigated. 
The impact of this phenomenon was recently highlighted in detail by a LAION team when training CLIP models \cite{2022laion} as well as a BigScience team during BLOOM training \cite{2022bloom}.

Several recent works have focused on fault-tolerant data-parallel training via dynamically changing the global batch size \cite{2021coddl,torchelastic, 2020elan, 2022aryl}. 
Fault-tolerant hybrid-parallel training is more challenging because the model is distributed across multiple GPUs. 
There are two primary approaches. 
\begin{denseenum}
    \item \emph{Checkpointing:} Checkpointing is a popular mechanism to persistently store training progress.
    For example, the BigScience team training the BLOOM model \cite{2022bloom} and the Meta AI team training the OPT model \cite{2022opt} used this recently. 
    However, manual reconfiguration after identifying and replacing the failed GPUs with spare ones is time-consuming.
    Varuna \cite{2022varuna} introduced job morphing to dynamically reconfigure training jobs to achieve the best performance with the remaining resources after restarting from the most recent checkpoint.
    While this does not introduce significant fixed overhead, recovery time can be unsustainable when the failure rate is high \cite{2023bamboo}.
    
    \item \emph{Redundant computation (RC):} To avoid reconfiguration and restart overheads, Bamboo \cite{2023bamboo} recently introduced redundant computation (RC) where each pipeline stage is redundantly computed in two subsequent nodes.
    When a node fails, the backup node computes the forward and backward passes of the failed node.
    RC introduces fixed computational overhead due to redundancy in computing and memory overhead of holding redundant states in each node.
    Note that reconfiguration and restart from a checkpoint is still necessary if two adjacent nodes fail.
    
\end{denseenum}

\subsection{Limitations of the State-of-the-Art}
\label{sec:limitations-of-sota}

\begin{figure}
    \centering
    \begin{subfigure}[t]{0.555\linewidth}
        \centering
        \includegraphics[width=\textwidth]{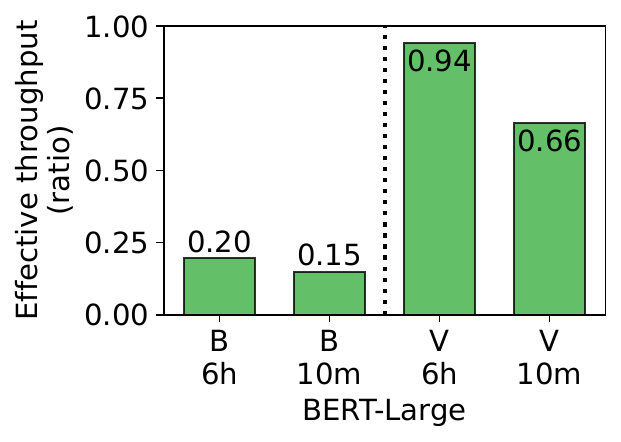}
    \end{subfigure}%
    \begin{subfigure}[t]{0.445\linewidth}
        \includegraphics[width=\linewidth]{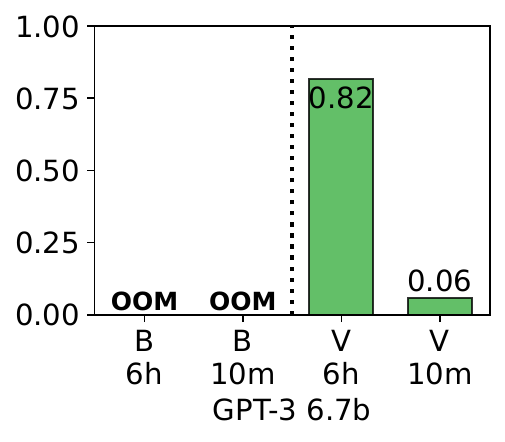}  
    \end{subfigure}
    \caption{Effective time spent in training for Bamboo (B) and Varuna (V) running BERT-large and GPT-3 6.7b models for different frequency of failures (6h and 10m).
        An optimistic upper-bound of the optimal is 1.00, when training remains unaffected by failure(s).
    }
    \label{fig:bamboo_varuna_tput}
\end{figure}

State-of-the-art approaches for fault-tolerant hybrid-parallel training do not provide any systematic fault tolerance guarantees, and they have large overheads, especially when models become larger. 

Figure~\ref{fig:bamboo_varuna_tput} shows the time effectively spent in training (\ie, the time that leads to training throughput) using Varuna and Bamboo when training BERT-Large and GPT3-6.7B -- with 340 million and 6.7 billion parameters, respectively -- when one failure happens every six hours and every 10 minutes on average. 
Detailed experimental setup is in Section~\ref{sec:exp_setup}.
Varuna provides higher training throughput compared to Bamboo, but it has noticeable job restart overhead.
Although some recent proposals improved checkpointing overhead \cite{2021checkfreq, 2022checknrun}, loading checkpoints upon restarts is still in the critical path.
When failures become more frequent, restarting overheads dominate.
Additionally, performance degradation after a failure is not proportional to the number of failures, but worse.
This is because Varuna's hybrid parallelism uses a grid topology; one GPU failure breaks the grid of GPUs, leaving some of them idle.

Bamboo, in contrast, reduces checkpointing and restart overheads, but RC in Bamboo introduces significant performance overhead, even when some portion of the overhead is hidden in pipeline bubbles.
Specifically, its forward RC redundantly computes the next stage all the time, lowering throughput even in the absence of failures.
Backward RC takes place only after failure(s), but it adds additional overhead to some pipelines' iteration times, making them stragglers and inflating the iteration time of synchronous training.
Worse, Bamboo also needs to restart with a full reconfiguration from a checkpoint for as few as two failures when two adjacent nodes fail.

Finally, both approaches perform poorly for larger models, especially when failures are frequent.
Bamboo runs out of memory and Varuna spends most of the time preparing to train.
We aim to design a solution that works well regardless of the frequency of failures both in terms of the fault tolerance guarantee it provides and the throughput it achieves. 

\section{\name Overview}
\label{sec:overview}

\name is a resilient distributed training platform for large models with guaranteed fault tolerance.
It presents the concept of \textit{pipeline templates} to achieve high throughput and fast fault tolerance at the same time (\S\ref{sec:pipeline_templates}).
Its fault tolerance guarantees allow for reconfiguration without restarts for up to $f$ simultaneous failures in the worst case (\S\ref{sec:guarantee_fault_tolerance}).
We also discuss \name's overall architecture (\S\ref{sec:system_components}) and how it integrates with the training lifecycle (\S\ref{sec:training_lifecycle}).



\subsection{Pipeline Templates}
\label{sec:pipeline_templates}
\name introduces \textit{pipeline templates}, each of which is a pipeline specification that defines how many nodes should be assigned to a pipeline, how many stages to create, and how to map model layers in stages to GPUs.
All pipelines \emph{instantiated} by \name are from precomputed pipeline templates.
In practice, \name instantiates multiple (possibly heterogeneous) pipelines from a set of heterogeneous pipeline templates to fully utilize an arbitrary number of nodes even when they do not form a grid.
Decoupling ``planning'' (pipeline template generation) from ``execution'' (pipeline instantiation) enables fast failure recovery; a pipeline with lost node(s) is replaced with a new pipeline instantiated from another pipeline template that requires a fewer number of nodes.


\begin{figure}[!t]
    \centering
    \subfloat{\includegraphics[width=\linewidth]{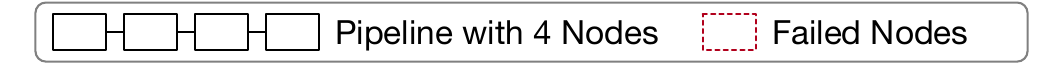}}
    \vspace{5pt}

    \setcounter{subfigure}{0}
    \subfloat[Failures happen to nodes with the same stage (S1 lost, not recoverable)]{\includegraphics[width=0.48\linewidth]{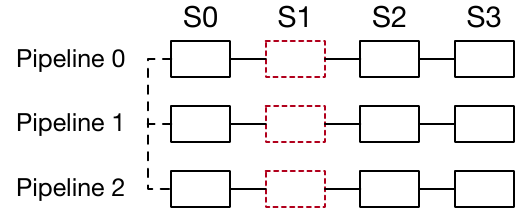}\label{fig:fault_tolerance_guarantee_a}}
    \hspace{0.01\linewidth}
    \subfloat[Failures in random places (all stages alive, recoverable)]{\includegraphics[width=0.48\linewidth]{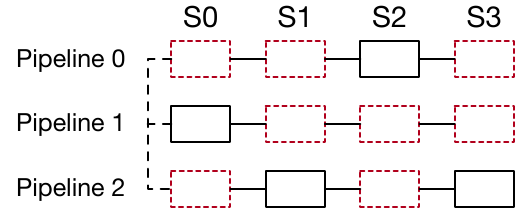}\label{fig:fault_tolerance_guarantee_b}}
    
    \caption{An example of \name's fault tolerance guarantees with $f=2$. S refers to a pipeline stage.
    (a) In the worst case, we lose model states (a stage) if more than $f$ nodes fail.
    (b) In the general case, however, more than $f$ node failures can be tolerated.}
    \label{fig:fault_tolerance_guarantee}
\end{figure}

\subsection{Fault Tolerance Guarantees}
\label{sec:guarantee_fault_tolerance}
\name guarantees fault tolerance without restart for up to $f$ simultaneous \emph{pipeline} failures, because in the worst case, $f$ node failures are enough to cause $f$ pipelines to fail.
Consider Figure~\ref{fig:fault_tolerance_guarantee}, where there are three pipeline replicas each with four stages -- \ie, each stage has three replicas.
We can tolerate at most two simultaneous node failures in the worst case, because if three failures take out all three replicas of any stage (stage 1 in Figure~\ref{fig:fault_tolerance_guarantee_a}), the pipelines cannot be recovered.
In general, however, \name can tolerate in excess of $f$ node failures, provided that a minimum of one copy of the entire model states is retained across the pipelines.
For example, even after eight node failures in Figure~\ref{fig:fault_tolerance_guarantee_b}, one copy of each stage still remains alive; hence, it is recoverable.




\subsection{System Components}
\label{sec:system_components}
\name extends existing ML training frameworks in two primary aspects (Figure~\ref{fig:overview}).
First, it has \textit{a pipeline template generator} to generate a set of heterogeneous pipeline templates that can be used by the execution engine for pipeline instantiation.
Pipeline templates are created only once and never change during the entire training.

Second, \name has a \textit{distributed execution engine} that enables efficient heterogeneous pipeline execution.
It instantiates pipelines from the given set of pipeline templates considering the user's fault tolerance threshold ($f$) and batch information (global batch and microbatch size).
It creates at least $f+1$ (possibly heterogeneous) pipeline replicas so that at least one copy of the model exists anytime during training for up to $f$ simultaneous failures.
The batch distributor calculates the number of microbatches for each pipeline that balances execution latency between heterogeneous pipelines.
Pipeline instantiation and batch distribution happens whenever a node fails or is added. 
The node change monitor detects node failure(s) and node additions; then the execution engine dynamically reconfigures using precomputed pipeline templates.

\begin{figure}[!t]
    \includegraphics[width=\linewidth]{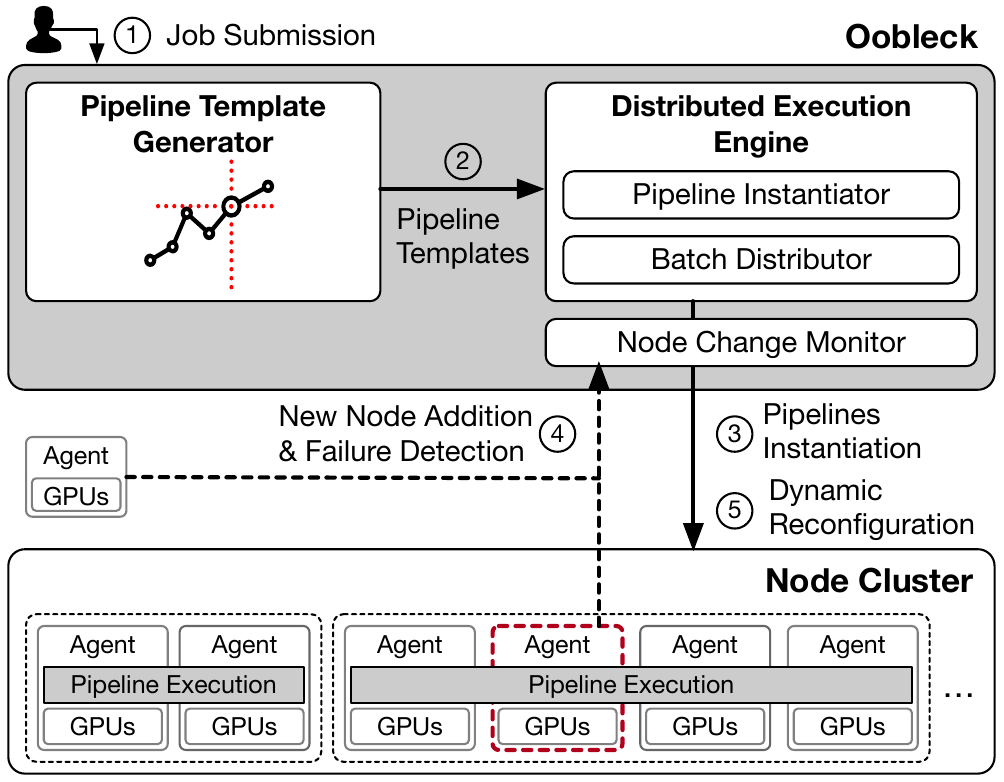}
    \caption{\name~system overview.}
    \label{fig:overview}
\end{figure}

\subsection{Training Lifecycle}
\label{sec:training_lifecycle}
\name users submit training jobs with a fault tolerance threshold $f$, a model and dataset to train on $N$ homogeneous nodes (received from a GPU cluster manager \cite{2019philly, 2022mlaaswild}), and batch size information \circled{1} (Figure~\ref{fig:overview}).
\name's pipeline template generator first creates a set of pipeline templates \circled{2}.
The distributed execution engine instantiates pipelines from the templates and deploys them on the cluster \circled{3}.

When node failure(s) happen \circled{4}, if we have a complete model replica, \name does not restart but reconfigures the pipelines \circled{5}.
The execution engine reinstantiates pipelines from the templates to make sure all nodes are used.
During pipeline reinstantiation, nodes share information about the ownership of model states and copy missing model states from others.
After reconfiguration and model states copying are done, nodes resume training.
A job runs until it reaches the target accuracy, a user terminates it, or \name cannot maintain $f+1$ pipeline replicas.
If the cluster cannot hold $f+1$ replicas, \name stores the progress, informs the user, and exits.
Thereafter, the user can decide to restart training from a recent checkpoint once enough nodes have recovered to maintain $f+1$ replicas.


\begin{figure*}

  \subfloat[Generating pipeline templates (\S\ref{sec:pipeline_template})]{\includegraphics[width=0.5\textwidth]{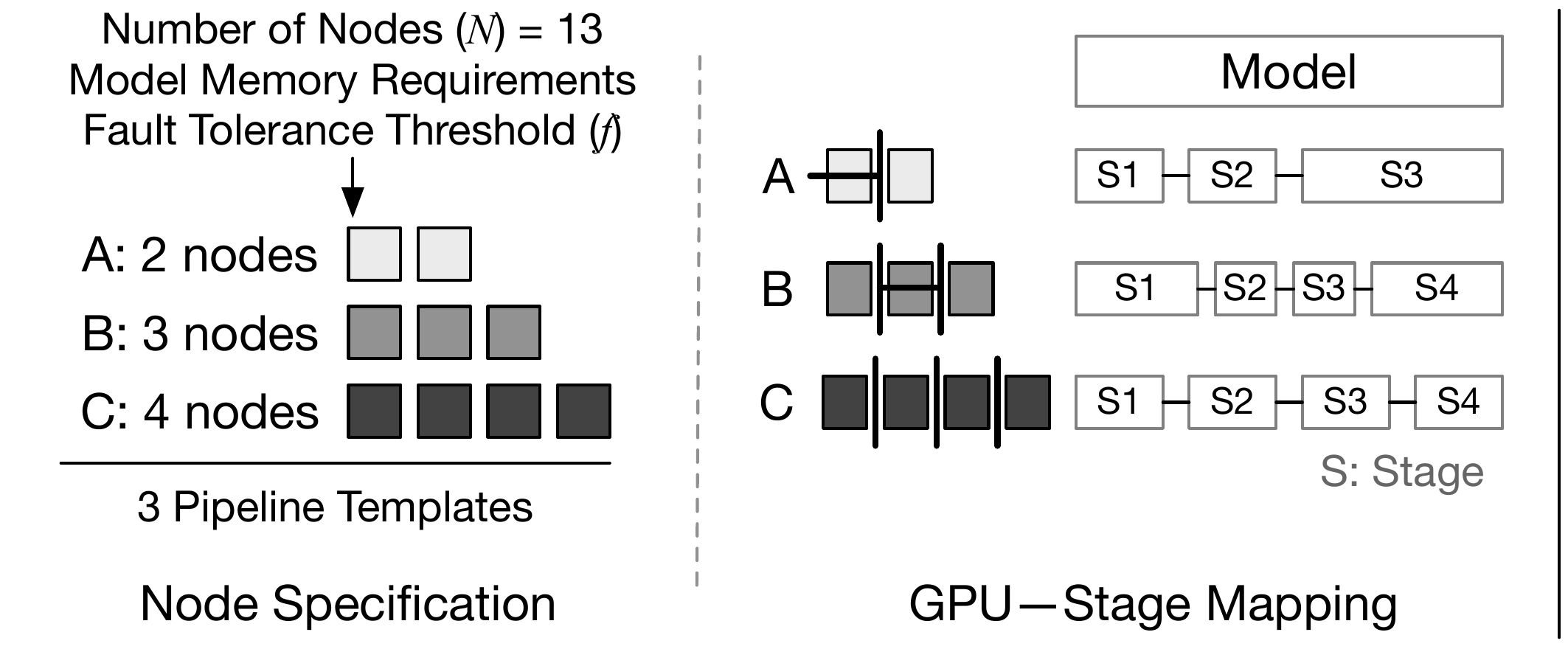}\label{fig:planning_overview_a}}
  \subfloat[Pipeline instantiation (\S\ref{sec:pipeline_instantiation})]{\includegraphics[width=0.5\textwidth]{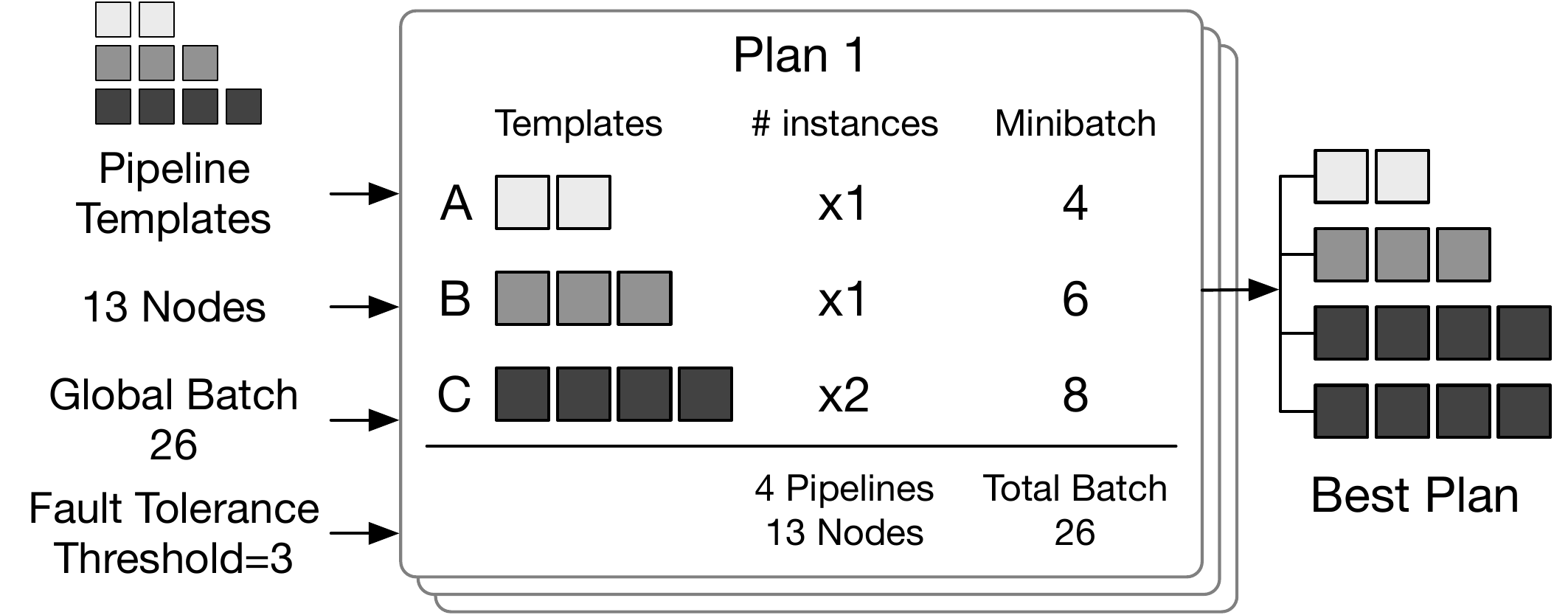}\label{fig:planning_overview_b}}

  \caption{\name's planning algorithm overview. 
    First, it generates a set of pipeline templates, a combination of which can utilize all available nodes.
    A template is a specification of pipelines, how many nodes are assigned and how GPUs in the nodes should be mapped to pipeline stages.
    Then, pipelines are instantiated following the fastest (best) plan after checking all possible plans.
    A plan includes how many pipelines should be instantiated from each pipeline template given the number of nodes and how batch size should be distributed to the pipelines.}
  \label{fig:planning_overview}
\end{figure*}

\section{\name Planning Algorithm}
\label{sec:planning}

\name tolerates $f$ simultaneous failures by instantiating $r$ ($\ge f+1$) heterogeneous pipeline replicas of the same model.
Each of these logically equivalent pipeline replicas performs hybrid-parallel training.
Unlike existing solutions that force a single homogeneous hybrid-parallel configuration over a rigid grid (\# GPUs per pipeline stage $\times$ \# pipeline stages $\times$ \# pipeline replicas) \cite{2022varuna}, \name's \textit{heterogeneous pipeline execution} can utilize all available GPUs.

Because the number of available nodes can vary over time due to failures, \name requires an effective mechanism to derive \emph{all} possible configurations of heterogeneous pipelines that can utilize all available GPUs at any point in time.
\name's pipeline template generator computes a fixed set of pipeline templates at the beginning of the training job for the entire training (\S\ref{sec:pipeline_template}).
The pipeline execution engine instantiates zero or more copies of each of the templates (\ie, a collection of heterogeneous pipeline replicas) to utilize all currently available nodes (\S\ref{sec:pipeline_instantiation}).

\subsection{Generating Pipeline Templates}
\label{sec:pipeline_template}
Each pipeline template created by \name is a set of specifications that defines how many nodes to use, and how the given GPUs and model layers are mapped to make pipeline stages use all those nodes.
Figure~\ref{fig:planning_overview_a} illustrates this process.
In this example, we generate a set of pipeline templates.
We first determine the number of heterogeneous pipeline templates and their node specifications (number of nodes) needed to utilize all available nodes, given the initial number of nodes $N$, the amount of memory required to train a model, and the fault tolerance threshold $f$ (\S\ref{sec:node_specification}).
In this case, three heterogeneous pipeline templates with 2, 3, and 4 nodes have been chosen.
Then, for each template, we partition the model and map the available GPUs to them to create pipeline stages that minimize the iteration time (\S\ref{sec:dc_algorithm}).

\subsubsection{Node Specification.}
\label{sec:node_specification}
Because pipeline templates never change during training and generating arbitrarily many of them is expensive, we must determine how many pipeline templates are needed, and then how many nodes each of the templates should use, so that some combination of them can always utilize any number of available nodes, even when we have fewer number of nodes than at the beginning after failures.

This can be formulated as the Frobenius problem \cite{2005frobenius}, which finds the Frobenius number $g$, the largest number that cannot be represented as a linear combination of integers.
Meaning, any number of available nodes after failures $N' > g$ can be expressed as a linear combination of the given pipeline templates, each with a specified integer number of nodes.

If we represent the number of pipeline templates as $p$ and the number of nodes for the $i$-th pipeline template be ordered values $n_i (0 < n_i \le N)$ where $n_i < n_{i+1}$, we can guarantee that any feasible $N' \ge (f+1) n_0$ is always larger than $g$ when the following conditions are met \cite{1956frobenius}. 
\begin{denseenum}
  \item $p > n_0 - 1$.
  \item $n_i$ are consecutive integers ($n_i+1=n_{i+1}$).
\end{denseenum}
See Appendix~\ref{sec:apdx_proof_frobenius} for a proof.

We set the lower bound of $N'$ as $(f+1)n_0$: the smallest number of nodes required to maintain $f+1$ replicas of the model, because $n_0$ is the smallest number of nodes for a single pipeline.
Any smaller $N'$ cannot respect the fault tolerance threshold $f$.

\parabf{Choice of $n_0$ and $p$.}
There are several choices for a set of pipeline templates that satisfy the conditions depending on the values of $n_0$ and $p$.
We choose the smallest possible $n_0$ and the largest $p$.
We select the minimum $n_0$ because shallow pipeline execution (smaller $n_0$) typically takes less time for the same amount of computation~\cite{2022turingnlg}.
Although a large $p$ does not directly benefit planning, it helps reduce reconfiguration overhead.
The largest $p$ can be calculated from the largest possible value for $n_{p-1}$ (referred to as $n^{\max}_{p-1}$).
When all but one of the $f+1$ replicas use $n_0$ nodes and the last one uses all the remaining nodes, $n^{\max}_{p-1} = N - fn_0$.
We now have $p$ to be the length of the range from $n_0$ to $n^{\max}_{p-1}$.

\subsubsection{GPU--Stage Mapping.}
\label{sec:dc_algorithm}
Given the number of nodes in each pipeline template, we must determine how to best use them by finding the number of pipeline stages, partitioning the model layers to the stages, and mapping the nodes to those stages.
%
We propose a divide and conquer algorithm to find the mapping that minimizes the iteration time.
This algorithm divides the model into pipeline stages and the nodes into a set of GPUs at the same time, and then maps each of them so that we utilize all GPUs in the given nodes.
It then iterates over all possible combinations of GPU--stage mapping and finds the one that minimizes the iteration time.

\begin{figure}
  \includegraphics[width=\linewidth]{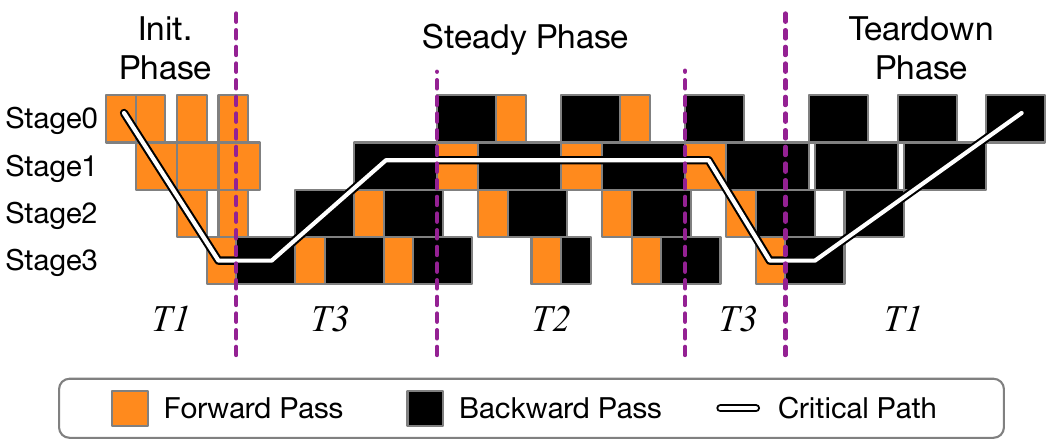}
  \caption{1F1B pipeline execution breakdown ($T1, T2, T3$)}
  \label{fig:pipeline_breakdown}
\end{figure}

Let $T(S', u, v, d)$ be the minimum iteration time for layers ($l_u, l_{u+1}, \dots, l_{v-1}$) partitioned into $S'$ stages and running on $d$ GPUs.
A pipeline for the entire model using all GPUs in the pipeline template then has the minimum iteration time $T(S, 0, L, n \cdot M)$, where the model has $L$ layers, and there are $n$ number of nodes in the pipeline template, each of which has $M$ GPUs.

To calculate the minimum iteration time of a pipeline, the algorithm considers its critical path and breaks $T$ down into three terms $T1$, $T2$, and $T3$ (Figure~\ref{fig:pipeline_breakdown}).
$T1$ represents 1F1B initialization and teardown phases of pipeline execution, which include one forward and one backward for all stages.
The steady phase in the middle has one forward and one backward pass alternating.
The critical path may still include forward and backward passes of other stages than the slowest one on each end, similar to $T1$.
We thus split the steady phase into $T2$ and $T3$.
$T2$ includes the slowest stage alternating forward and backward, and $T3$ is the remaining part.

\parabf{Divide.}
In the division phase, we divide the model and the nodes at the same time.
Division continues until we cannot partition either GPUs or model layers, or the number of partitions matches the desired number of stages.
If multiple GPUs are assigned to a pipeline stage, tensor parallelism is used to accelerate it.
Figure~\ref{fig:dc_algorithm} illustrates such a division and mapping process.
After both sub-problems are conquered, the algorithm combines their results to calculate the execution time of a multi-stage pipeline created by connecting two sub-problems.
From the definitions of $T1$, $T2$, and $T3$, the division and combination process can be defined as recursive structures for each term:
\begin{flalign}
  \begin{split}
    &T1_{s, k, m}(S', u, v, d) \\
    &=T1_{s, k, m}(s, u, k, m) + T1_{s, k, m}(S'-s, k+1, v, d-m)
  \end{split}&& \label{eq:T1dc}
\end{flalign}
    
\begin{flalign}
  \begin{split}
    &T2_{s, k, m}(S', u, v, d~~\vert~~k^*) \\
    &= (N_b - S' + k^* - 1)(F_{s_{k^*,m}} + B_{s_{k^*,m}})
  \end{split}&& \label{eq:T2dc}
\end{flalign}
    
\begin{flalign}
  \begin{split}
    &T3_{s, k, m}(S', u, v, d~~\vert~~k^*)\\
    &= \begin{cases}
      \begin{aligned}
         &T3_{s, k, m}(s, u, k, m~~\vert~~k_1^*) \\ 
        +&T1_{s, k, m}(S'-s, k+1, v, d-m)
      \end{aligned} & \text{if } k^* == k_1^* \\[1em]
      T3_{s, k, m}(S'-s, k+1, v, d-m~~\vert~~k_2^*) & \text{else}
    \end{cases}
  \end{split}&& \label{eq:T3dc}
\end{flalign}
We iterate over $s$, $k$, and $m$ globally, and find a ($s, k, m$) that minimizes $T1_{s, k, m}+T2_{s, k, m}+T3_{s, k, m}$.
Each $T1_{s, k, m}$, $T2_{s, k, m}$, and $T3_{s, k, m}$ is the solution of $T1$, $T2$, and $T3$, respectively.

$k^*$ denotes the index of the slowest stage, derived from either $k_1^*$ or $k_2^*$, the slowest stage indices of the two sub-problems.
$T2$ depends on the number of microbatches ($N_b$) deployed to the pipeline, which is not yet determined.
From prior observations that the pipeline bubble overhead is negligible with $N_b \ge 4S'$ \cite{2019gpipe}, we temporarily use $N_b=4S'$ in planning.
The structure of $T3$ is special as it includes $T1$ in Equation~\ref{eq:T3dc}.
T3 is an accumulation of forward and backward time for all the following stages after the slowest ($\sum_{k=k^*}^{S'-1} (F + B)$).
If $s_{k^*}$ is in the first half sub-problem (\ie~$s_{k^*} == s_{k_1^*}$), it can be broken down to ${\sum_{k=k^*}^{s}(F+B)}+{\sum_{k=s+1}^{S'-1}(F+B)}$, each of which represents $T3$ of the first half and $T1$ of the second half, respectively.


\begin{figure}
  \includegraphics[width=\linewidth]{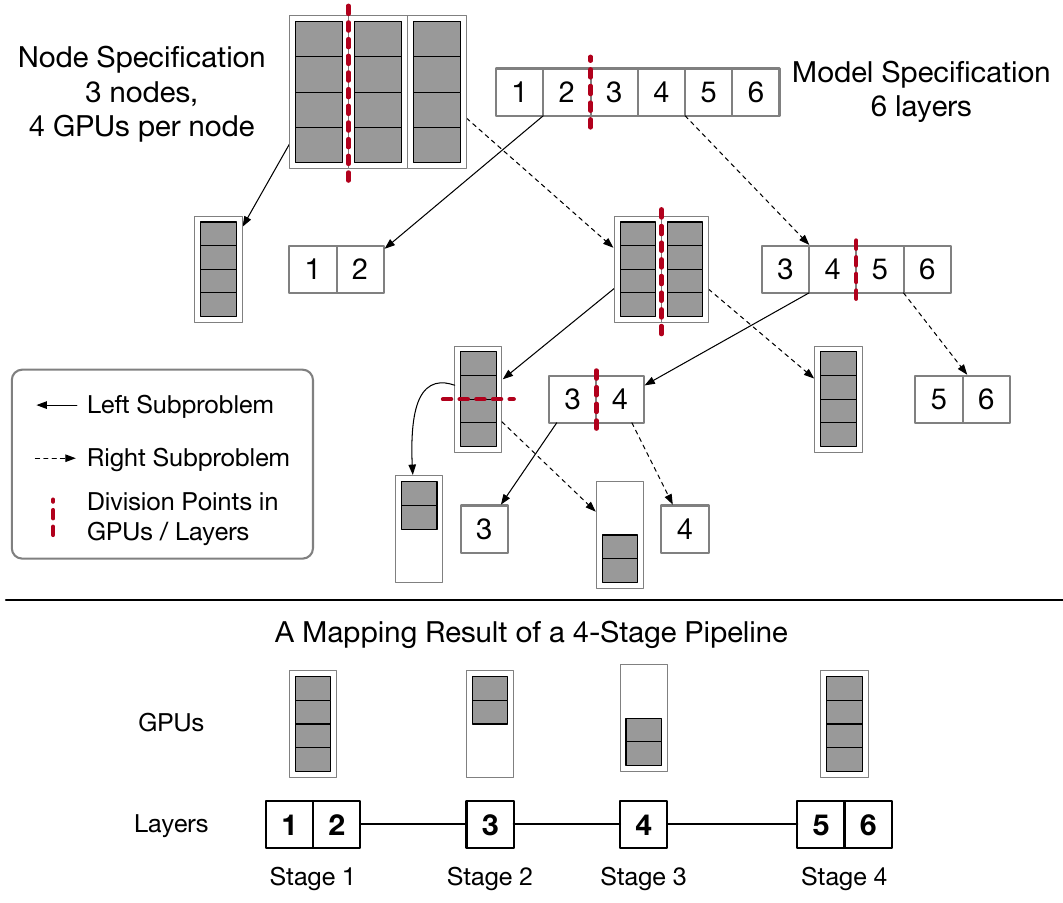}
  \caption{A toy example of division process for a 4-stage pipeline template with 3 nodes (template B in Figure~\ref{fig:planning_overview}) and a model with 6 layers.
  The model and the GPUs in nodes are divided into two sub-problems together. When division is done, each group of partitioned GPUs and layers form a stage.
  The algorithm iterates all combinations of layer partitioning and GPU partitioning to find the minimum $T$.}
  \label{fig:dc_algorithm}
\end{figure}

\parabf{Conquer.}
When a problem has just one stage, we can easily calculate the execution time of running a stage $s$ with $l_u, \dots, l_{v-1}$ layers on $d$ GPUs:
\begin{equation}
\begin{aligned}
T1(1, u, v, d) =& F_{s, d} + B_{s, d} = \sum_{k=u}^{v-1} ( F_{l_k, d} + B_{l_k, d} ) \\
T2(1, u, v, d) =& 2 ( F_{s, d} + B_{s, d} ) \\
T3(1, u, v, d) =& F_{s, d} + B_{s, d}
\end{aligned}
\label{eq:dc_init}
\end{equation}
There is one requirement for $d$ GPUs running a single stage: \emph{all $d$ GPUs should be in the same node}.
It is reasonable because if GPUs span several nodes, the cross-node network becomes a bottleneck and lowers the utilization of high-throughput intra-node network in collective communications done in intra-layer parallel execution.
We simply mark all $T$ values as $\infty$ if cross-node GPUs are given.

\parabf{Choosing the number of stages $S$.}
We do not know which $S$ provides the minimum iteration time.
Therefore, we iterate over possible values of $S$ in $(n, n+1, \dots, L)$.
Because we partition the model at layer granularity, the number of stages cannot exceed the number of layers $L$.
The minimum is derived from the constraint that a single stage cannot be assigned to two or more nodes.
If $S$ becomes less than $n$, it breaks the constraint and some stages should have at least two nodes assigned according to the pigeonhole principle.

\parabf{Time complexity of the naive implementation.}
The recursive stage division happens in $O(L)$.
For every division, stages and layers are partitioned and they are assigned to two device sub-clusters.
Stage and layer partitioning have $O(L)$ choices.
Partitioning nodes is done in $O(n)$, but GPUs within a single node can further be partitioned, adding $O(M)$.
Time complexity of layer partitioning and device assignment for a given number of stage is $O(LnM)$.
Divide and conquer happens for each feasible $S$, and iterating over $S$ values is $O(L-n)$.
Therefore, the overall algorithm time complexity per pipeline template is $O\left((L-n)L^3nM\right)$.

\parabf{Using memoization to reduce complexity.}
We cache all intermediate results to accelerate the divide and conquer algorithm.
It boosts not only getting the mapping of one pipeline template but also helps in deriving the mapping of the other pipeline templates.
In fact, running the algorithm for the largest pipeline template (with $n_{p-1}$ nodes) is enough to calculate intermediate caches required for building the mapping of all the other pipeline templates.
With all intermediate caches present, calculating the mapping of another smaller pipeline template can be done in $O(Ln)$.

\subsection{Pipeline Instantiation}
\label{sec:pipeline_instantiation}
Given a set of pipeline templates, we know by construction (\S\ref{sec:node_specification}) that there exists a combination of them that \name can instantiate to utilize all available nodes.
However, such a combination of heterogeneous pipelines may not be unique.
So we first find all such feasible combinations (\S\ref{sec:enumerating_instantions}).
\name chooses a plan with the highest estimated throughput among all the feasible combinations (\S\ref{sec:batch_distribution}).


\begin{figure}
  \includegraphics[width=\linewidth]{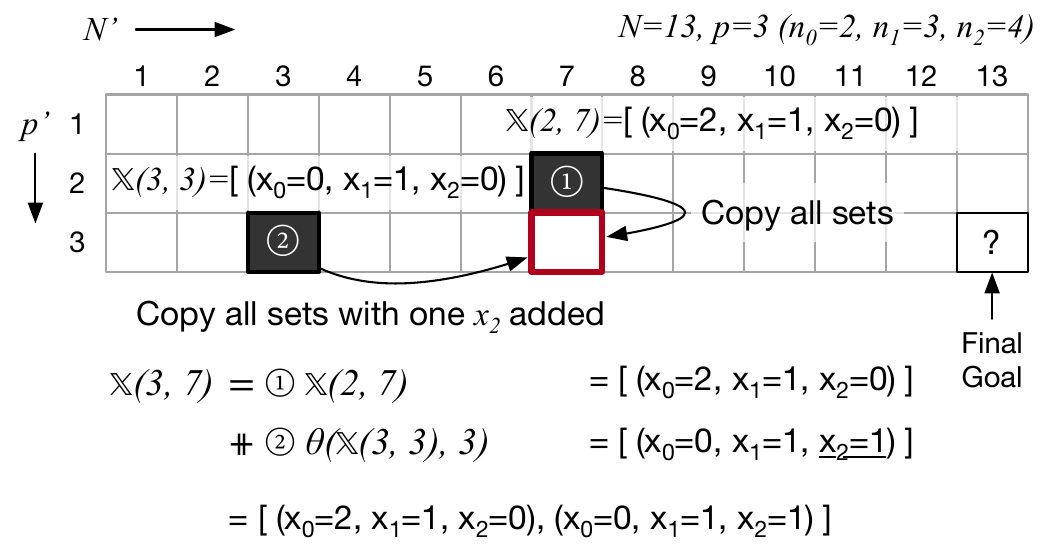}
  \caption{The dynamic programming algorithm finding all list of feasible $\mathbf{X}$s. Underlined $x_2=1$ is added by $\theta()$ function.}
  \label{fig:instantiation_dyp}
\end{figure}

\subsubsection{Enumerating All Instantiation Options.}
\label{sec:enumerating_instantions}
Although we know that we can utilize all available nodes with the set of pipeline templates, the number of pipelines to be instantiated from each pipeline template is undetermined.
Worse, there may be several pipeline configurations that use all nodes from the same pipeline template set.
For example, 13 nodes can be represented as the plan 1 in Figure~\ref{fig:planning_overview_b} $(1 \cdot n_0 + 1\cdot n_1 + 2 \cdot n_2)$, but also as $(0 \cdot n_0 + 3 \cdot n_1 + 1 \cdot n_2)$, and more.
We therefore enumerate all feasible pipeline sets for currently available nodes and pick one that maximizes training throughput.

Let $\mathbb{X}(p, N)$ be a list of all feasible pipeline sets $[\mathbf{X}_0, \mathbf{X}_1, ...]$.
Each $\mathbf{X}_i$ is a set of the number of pipelines to be instantiated $(x_0, x_1, \dots, x_{p-1})$, using $p$ number of heterogeneous pipeline templates with the number of nodes specification ($n_0, n_1, \dots, n_{p-1}$), so that all $N$ nodes are used by pipelines.
A feasible $\mathbf{X}_i$ satisfies the following requirements:
\begin{denseenum}
  \item \label{item:req1} $N = x_0n_0 + x_1n_1 + \dots + x_{p-1}n_{p-1}$ (All nodes are used).
  \item \label{item:req2} $\sum_{j=0}^{p-1}x_j \ge f+1$ (Number of pipelines is at least $f+1$).
\end{denseenum}

We exploit dynamic programming for the coin change problem to find $\mathbb{X}$.
The coin change problem finds a combination of coins that add up to the given amount of money~\cite{coin_change}.
It is an equivalent problem to Requirement~\ref{item:req1} above if we replace denominations of each coin with $n_i$, and the given amount of money with $N$.
We formulate the dynamic programming structure as:
\begin{equation}
  \mathbb{X}(p', N') = \mathbb{X}(p'-1, N') \dplus \theta(\mathbb{X}(p', N'-n_{p'}), p')
  \label{eq:instantiation_dp}
\end{equation}
where $\dplus$ means concatenating two lists, and $\theta(\mathbb{X}, p')$ is a function that increases $x_{p'}$ by 1 in every $\mathbf{X}_i$s in $\mathbb{X}$.

Figure~\ref{fig:instantiation_dyp} shows the execution of the dynamic programming algorithm.
The two terms in Equation~\ref{eq:instantiation_dp} are associated with each black boxes.
$\mathbb{X}$ in the red box should include all $\mathbf{X}_i$s that use all seven nodes in instantiating pipelines using the three different pipeline templates ($n_0=2, n_1=3, n_2=4$).
$\mathbb{X}$ in \circled{1}, all $\mathbf{X}_i$s use seven nodes and are already feasible for $\mathbb{X}(3, 7)$. We just copy them.
$\mathbb{X}$ in \circled{2}, however, only uses three nodes in total.
By adding one four-node pipeline (increasing $x_2$ by 1), all $\mathbf{X}_i$s use seven nodes and become feasible for $\mathbb{X}(3, 7)$.

The dynamic programming is done in $O(Np)$ filling all table elements. 
$\mathbb{X}$ in the bottom-right corner of the table contains all feasible $\mathbf{X}_i$s.
To satisfy Requirement~\ref{item:req2}, we filter the list and obtain sets with $\sum_{j=0}^{p-1}x_j \ge f+1$.

\subsubsection{Calculating Throughput with Batch Distribution.}
\label{sec:batch_distribution}
\name's execution engine needs to choose from several feasible $\mathbf{X}_i$s.
We calculate the overall throughput for each $\mathbf{X}_i$ and choose the one that maximizes the throughput.
To calculate throughput, we need to determine the batch size of each pipeline.
While the global batch size is given by the user, it is \name's responsibility to distribute them across heterogeneous pipelines to maximize overall throughput.
It is crucial to assign work proportional to the amount of computing power of each pipeline; otherwise, the overall throughput will be decreased due to stragglers.
We refer to this as \textit{batch distribution}.
Given the global batch size $B$ and microbatch size $b$, batch distribution calculates the number of microbatches for each pipeline that minimizes stragglers.

Let $N_{b, i}$ be the number of microbatches for $i$-th pipeline ($0 \le i < x,~x = \sum_{j=0}^{p-1}x_j$) and $T_i$ be the iteration time of the pipeline with a single microbatch of size $b$.
Minibatch size for $i$-th pipeline can be calculated as $N_{b, i} \times b$.
By adjusting $N_{b, i}$, we minimize variance between different pipelines' batch processing times. 
We formulate it as an integer optimization problem:
\begin{mini}
  {}{\sum_{i=0}^{x-1}(N_{b, i}T_i - \overline{N_{b}T})^2}{}{}
  \addConstraint{\sum_{i=0}^{p-1}{N_{b, i}bx_i} = B}
  \addConstraint{N_{b, i} \in \mathbb{N}}
  \label{eq:nonlinear_optimization}
\end{mini}
where $\overline{N_bT}$ is the average iteration time of all ($0 \le i < x$) pipelines.
Any integer nonlinear optimization solver can be used to get $N_{b, i}$ and thus minibatch size for each pipeline.

Note that the optimization may fail to redistribute batch properly, primarily when the global batch size is too small and cannot be split to integers.
\name does not change the global batch size arbitrarily in such cases. 
Instead, it recommends an adjusted global batch size close to the original one but distributable.

\section{Dynamic Reconfiguration}
\label{sec:dynamic_reconfiguration}

\begin{figure}[!t]
    \centering

    \subfloat[A node failure in a 4-node pipeline. We have a 3-node pipeline template, thus a new pipeline with 3 nodes is instantiated, which replaces the existing one.]{\includegraphics[width=\linewidth]{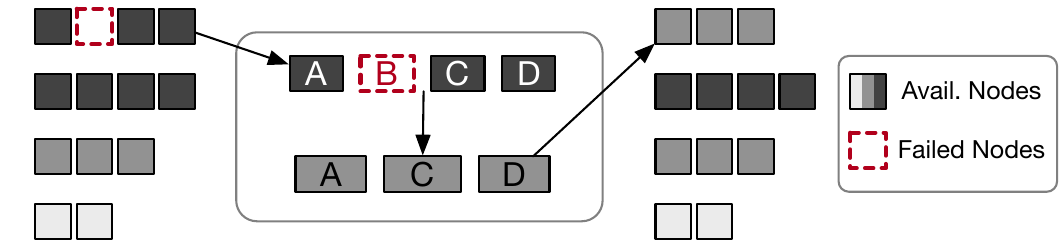}\label{fig:dynamic_reconfiguration1}}
    
    \subfloat[A node failure in a 2-node pipeline. Since there is no template for one node, it gets another node from another pipeline to keep the 2-node pipeline. Two affected pipelines reinstantiate or reconfigure themselves.]{\includegraphics[width=\linewidth]{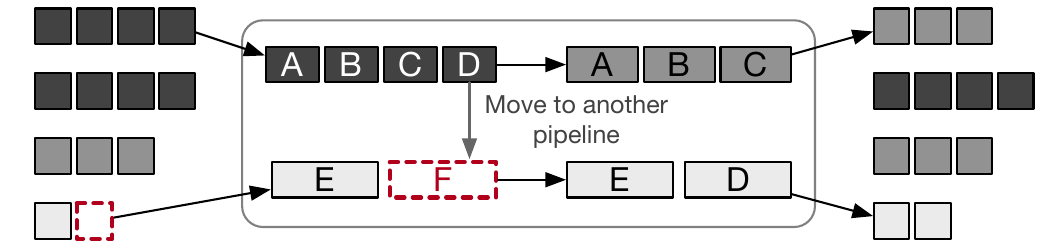}\label{fig:dynamic_reconfiguration2}}

    \subfloat[A node failure in a 2-node pipeline. Because it cannot borrow a node from any other pipeline, it is merged with another pipeline.]{\includegraphics[width=\linewidth]{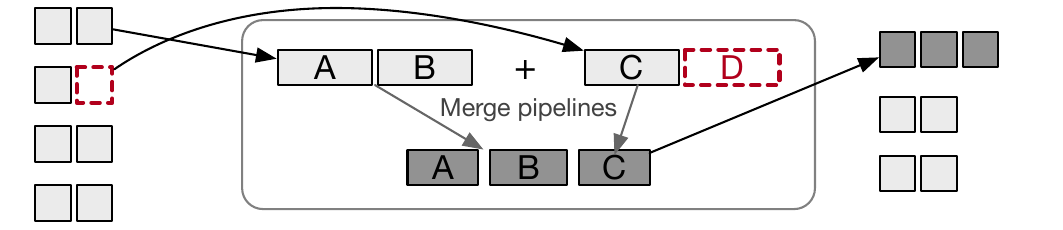}\label{fig:dynamic_reconfiguration3}}

    \caption{Three steps of pipeline reinstantiation. After reinstantiation is done, missing layers in a new pipeline are copied from other pipelines.}
    \label{fig:dynamic_reconfiguration}
\end{figure}

Upon a node failure, the pipeline it was assigned to becomes incomplete and has missing model states; therefore, training halts in that pipeline.
Pipelines affected by failures are replaced with new pipelines created via pipeline reinstantiation using precomputed pipeline templates (\S\ref{sec:reinstantiation}).
After reinstantiating pipelines, the nodes copy missing layers from unaffected pipeline replicas.
\name also redistributes batch in response to the pipeline configuration change (\S\ref{sec:redistribution}).
Provided that we have copies of the model states, \name can recover from failures until we have fewer than $(f+1)n_0$ nodes.

\subsection{Pipeline Reinstantiation}
\label{sec:reinstantiation}
\name instantiates a new pipeline from one of the pipeline templates, replacing the existing one affected by failures.
Given our limited number of pipeline templates, there might not be a suitable pipeline template for the remaining number of nodes.
Thus, pipeline reinstantiation is done in three steps:
simple reinstantiation, borrowing nodes, and merging pipelines.

For each pipeline, \name first checks if there is an instantiable pipeline template with remaining nodes; if so, \name simply reinstantiates it and replaces the old one (Figure~\ref{fig:dynamic_reconfiguration1}).
If there is no instantiable pipeline template with remaining nodes, \name tries to \textit{borrow nodes} from other pipelines until we have enough nodes to instantiate the smallest pipeline template (Figure~\ref{fig:dynamic_reconfiguration2}).
Pipelines that yield their nodes should also be reinstantiated with fewer nodes.

After many reconfigurations and node borrowings, all pipelines may not be able to yield their nodes.
When failures happen at this moment, the pipeline affected by failures cannot be reinstantiated due to a lack of nodes.
In such a case, \name \textit{merges pipelines} to create a bigger pipeline (Figure~\ref{fig:dynamic_reconfiguration3}).
It is guaranteed that we have an instantiable pipeline template for a merged pipeline if it has at least $n_0$ nodes (the minimum number of nodes to maintain one single pipeline). See Appendix~\ref{sec:apdx_pipeline_merge} for a proof.

\subsection{Batch Redistribution}
\label{sec:redistribution}
After pipeline reinstantiation, execution configuration has been changed and distributed batches no longer ensure balanced execution.
\name runs Equation~\ref{eq:nonlinear_optimization} again given the current set of the number of pipeline instances and continues training with a newly calculated batch size.
Each pipeline may have more batches to compute, but the global batch size remains constant.

\section{Implementation}
\label{sec:implementation}
We implement \name in Python using PyTorch~\cite{2020pytorch} and HuggingFace Transformers~\cite{2022hftransformer} using components from Merak~\cite{2023merak} in using PyTorch fx symbolic tracer~\cite{2022torchfx} to create pipelines and from DeepSpeed~\cite{2020deepspeed} to run them.
For the implementation of 3D parallelism, \name has integrated PyTorch Fully Sharded Data Parallel (FSDP)~\cite{2023fsdp} as a replacement for tensor parallelism in each stage~\cite{2020zero}, pipeline parallelism, and data parallelism.
We have used Pyomo library and its Mixed-Integer Nonlinear Decomposition Toolbox (MindtPy) solver for non-linear integer optimization~\cite{2011pyomo, 2021pyomo}.
Here, we elaborate on the key challenges addressed in implementing \name.


\begin{figure}[!t]
    \includegraphics[width=\linewidth]{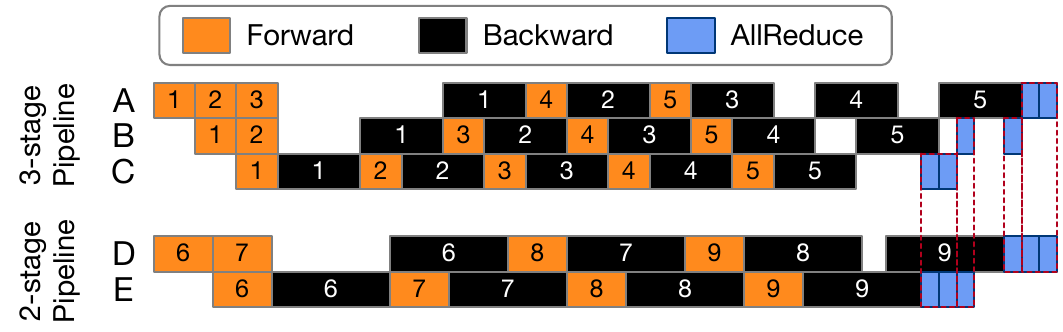}
    \caption{Heterogeneous pipeline execution with one 3-stage pipeline and one 2-stage pipeline.
    Allreduce synchronization happens at the end of iteration and done in layer granularity to communicate between multiple GPUs.}
    \label{fig:heterogeneous_pipeline}
  \end{figure}

\subsection{Model Synchronization Between Heterogeneous Pipelines}
\label{sec:heterogeneous_execution}
Model gradient synchronization between pipelines in hybrid parallelism is typically done at pipeline stage granularity.
However, because heterogeneous pipelines in \name have different stage configurations, stage-wise synchronization does not work.
\name instead breaks down stages into layers and synchronizes them individually, similar to PyTorch bucketing~\cite{2020pytorch}.
Figure~\ref{fig:heterogeneous_pipeline} illustrates an example of a 6-layer model execution with two heterogeneous pipelines.
Stage E in the 2-stage pipeline has 3 layers to synchronize which are stored in stage B and C in the 3-stage pipeline.
\name performs synchronization for each individual layer with potentially different peer nodes.

Data synchronization in smaller data might have performance issue because it might not fully saturate the network.
We overlap communication with computation to offset increased communication latency~\cite{2020pytorch}.

\subsection{Detecting Node Failures}
\label{sec:detecting_node_failures}
\name uses NCCL for communication between GPUs.
However, NCCL cannot detect unexpected communication channel disconnection and hangs when tries to communicate with the failed node until a timer expires.
To detect a node failure immediately, we launch a CPU process on each node and establish a TCP connection to a centralized CPU process.
When a node dies, a socket disconnection event is triggered and broadcasted for reconfiguration.

\section{Evaluation}
\label{sec:evaluation}

We evaluate the effectiveness of \name on large DNN models with 340M to 6.7B parameters and compare it against both Bamboo and Varuna.
We summarize the results as follows:

\begin{denseitemize}
    \item \name outperforms the state-of-the-art solutions by up to $29.6 \times$ when nodes fail more frequently and matches them as failures become less frequent (\S\ref{sec:tp_failure_freqs}).
    
    \item \name's benefits extend to real-world settings where nodes are out and join back following spot instance traces.
        It outperforms the rest by up to $9.1 \times$ on average (\S\ref{sec:trace_driven_tp}).

    \item Ablation studies show that \name's one-time planning overhead is low and it has high GPU utilization (\S\ref{sec:ablation}).
\end{denseitemize}

\subsection{Experimental Setup}
\label{sec:exp_setup}

\parabf{Cluster setup.}
We evaluate \name using 30 NVIDIA A40 GPUs with 40GB GPU memory each.
The GPUs are connected to each other via a 200Gbps Mellanox ConnectX-6 InfiniBand adaptor for communication.

Varuna requires a remote object storage to store checkpoints for fault tolerance.
We deploy a distributed object storage that consists of 6 nodes with two Intel Xeon Gold 6330 CPUs with 28 cores each, 512GB CPU memory, a 4TB PCIe 4.0 NVMe drive, and a 200Gbps Mellanox ConnectX-6 InfiniBand adaptor, respectively.
We use MinIO for distributed object storage software~\cite{2023minio}.

\begin{table}[!t]
    \footnotesize
    \caption{Model and batch configurations. * in Bamboo indicates the largest possible microbatch runnable in our evaluation environment. X means not runnable even with 1 microbatch size.}
    \label{tab:models}
    \begin{tabular}{c@{\hspace{-1pt}}c@{\hspace{8pt}}c|ccc}
        \toprule
               \multirow{2}{*}{} & \multirow{2}{*}{\# Params} & \multirow{2}{*}{\shortstack{Global\\Batch}} & \multicolumn{3}{c}{Microbatch Size} \\ 
                                 &                            &                               & Bamboo     & Varuna     & Oobleck     \\ \midrule
    BERT-Large~\cite{2019bert}   & 340M                       & 8192                          & 4*         & 32         & 32          \\
    GPT-2~~\cite{2019gpt2}       & 345M                       & 8192                          & 1*         & 32         & 32          \\ \midrule
    GPT-3 Medium~\cite{2020gpt3} & 350M                       & 8192                          & X          & 16         & 16          \\
    GPT-3 2.7b~\cite{2020gpt3}   & 2.7B                       & 1024                          & X          & 2          & 2           \\
    GPT-3 6.7b~\cite{2020gpt3}   & 6.7B                       & 1024                          & X          & 2          & 2           \\
    \bottomrule
    \end{tabular}
\end{table}

\begin{table*}[t]
    \small
    \caption{Throughput (samples/s) with different frequency of failures. Bamboo was not able to run any GPT-3 model due to lack of memory (OOM).}
    \label{tab:tput_freq}
    \begin{tabular}{l|rrr|rrr|rrr|rrr|rrr}
    \toprule
    Models           & \multicolumn{3}{c|}{BERT-Large} & \multicolumn{3}{c|}{GPT-2} & \multicolumn{3}{c|}{GPT-3 Medium} & \multicolumn{3}{c|}{GPT-3 2.7b} & \multicolumn{3}{c}{GPT-3 6.7b}    \\ \midrule
    Failure Frequency & 6h     & 1h     & 10m    & 6h    & 1h    & 10m   & 6h    & 1h    & 10m   & 6h   & 1h   & 10m  & 6h   & 1h   & 10m   \\ \midrule
    Bamboo           & 77.04  & 75.60  & 69.84  & 17.47 & 17.13 & 16.01 &       &       &       &      &      &      &      &      &       \\
    Varuna           & 259.57 & 245.39 & 168.15 & 86.42 & 83.94 & 69.67 & 29.52 & 27.61 & 20.71 & 7.27 & 6.41 & 0.36 & 4.02 & 2.91 & 0.12  \\
    Oobleck          & 287.10 & 286.28 & 282.11 & 85.59 & 85.42 & 84.80 & 29.30 & 29.21 & 28.70 & 7.29 & 7.23 & 6.89 & 4.33 & 4.22 & 3.55  \\
    \bottomrule
    \end{tabular}
\end{table*}

\parabf{Baselines.}
We compare \name to the following baselines:
\begin{denseitemize}
    \item \textit{Varuna}~\cite{2022varuna}: A resilient training framework based on automated parallel configuration and checkpoints~\cite{2022varunagithub}.
    \item \textit{Bamboo}~\cite{2023bamboo}: A resilient training framework based on redundant computation without full restart~\cite{2023bamboogithub}.
\end{denseitemize}
Neither of them supports 3D parallelism; hence, \name uses one GPU per node configuration to avoid its planner generating plans that cannot be implemented in our baselines.

Both works focus on utilizing spot instances, while \name supports general fault tolerance including preemptions in spot instance environments.
Spot instance environments have a unique mechanism of \textit{preemption notification}; the system is notified prior to actual preemption happening. 
In our spot instance-based evaluation, all three frameworks leverage early warning. 
In general, however, there is no such notification.

For Varuna, we periodically perform synchronous checkpointing for every 10 iterations following their continuous checkpointing policy~\cite{2022varuna}.

\parabf{Workloads.}
Table~\ref{tab:models} lists model configurations.
We adopt GPT-2 and BERT-Large from Bamboo and Varuna, and add three different configurations of GPT-3 from OpenAI~\cite{2020gpt3} to verify its scalability.
Although our evaluation only shows transformer models for comparison against two predecessors, \name's design is not limited to transformer language models and can support other DNN models.
For all evaluations, we use the Wikitext dataset~\cite{2016wikitext} and TF32 precision.

We also list batch size configurations for each framework in Table~\ref{tab:models}.
The reasons behind the discrepancy in batch sizes are twofold.
First, Bamboo needs to store additional model states for redundant computation, requiring $2\times$ memory.
Second, Bamboo does not use activation checkpointing,\footnote{This is because Bamboo's design choice stems from imbalanced memory consumption due to different amount of activations across stages.
Activation checkpointing~\cite{2016activationrecomputation} drastically reduces memory consumption by activations and it conflicts with Bamboo's design.}
while Varuna and Oobleck do.

\begin{figure*}
    \centering
    \begin{subfigure}[r]{\textwidth}
        \centering
        \includegraphics[width=0.4\textwidth]{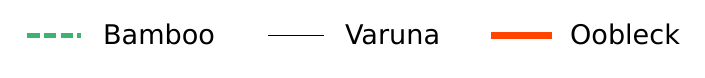}
    \end{subfigure}
    \begin{subfigure}[t]{0.25\textwidth}
        \includegraphics[width=\textwidth]{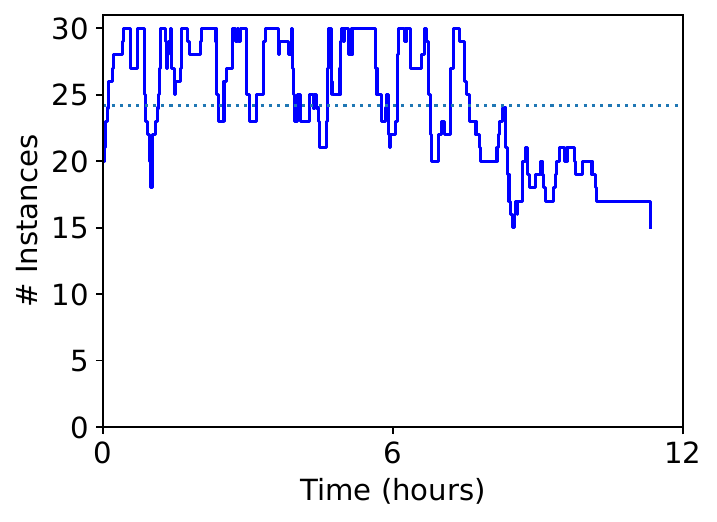}
        \label{fig:tpu_p3_node}
    \end{subfigure}%
    \begin{subfigure}[t]{0.25\textwidth}
        \includegraphics[width=\textwidth]{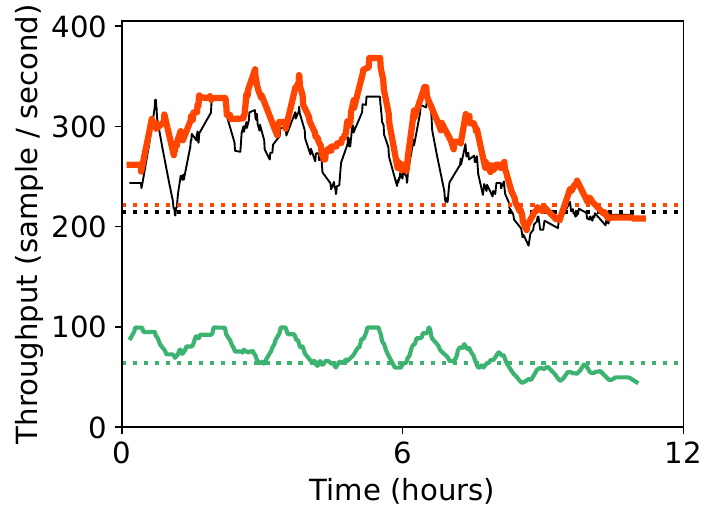}
        \label{fig:tput_p3_bert}
    \end{subfigure}%
    \begin{subfigure}[t]{0.25\textwidth}
        \includegraphics[width=\textwidth]{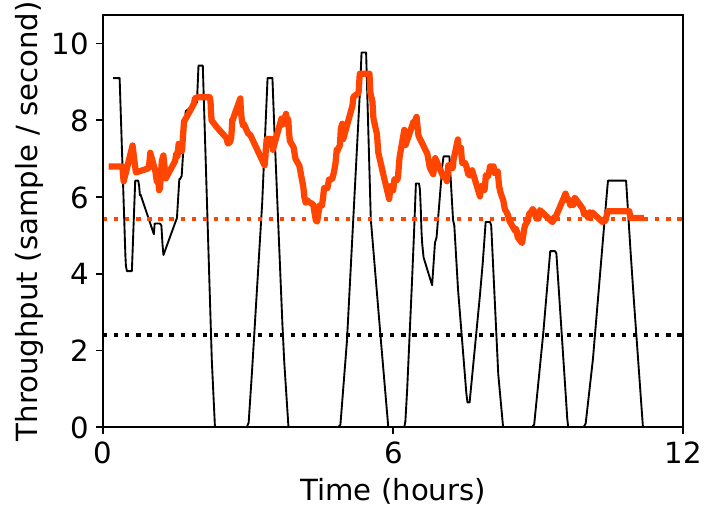}
        \label{fig:tput_p3_gpt327b}
    \end{subfigure}%
    \begin{subfigure}[t]{0.25\textwidth}
        \includegraphics[width=\textwidth]{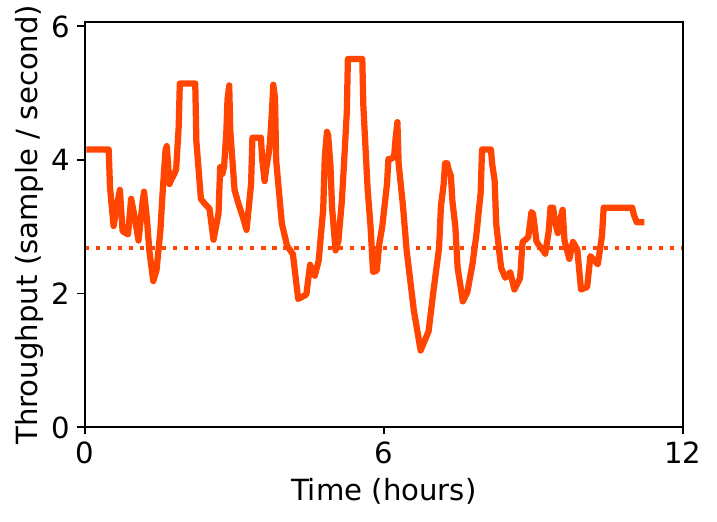}
        \label{fig:tput_p3_gpt367b}
    \end{subfigure}
    \begin{subfigure}[t]{0.25\textwidth}
        \includegraphics[width=\textwidth]{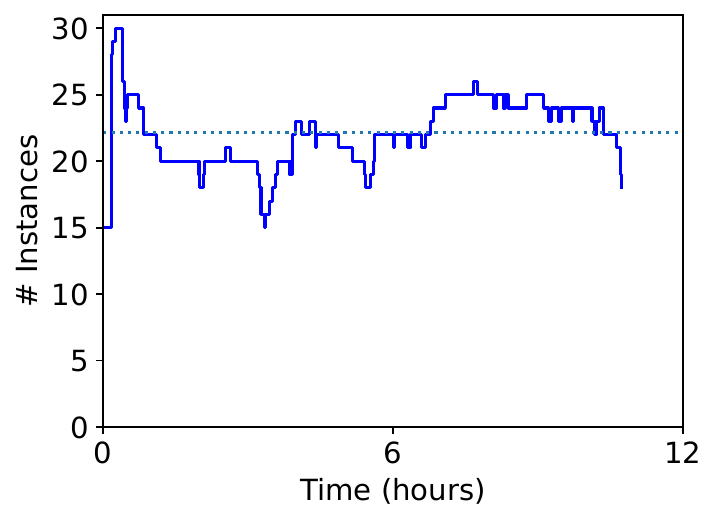}
        \caption{Node availability}
        \label{fig:tput_gcp_node}
    \end{subfigure}%
    \begin{subfigure}[t]{0.25\textwidth}
        \includegraphics[width=\textwidth]{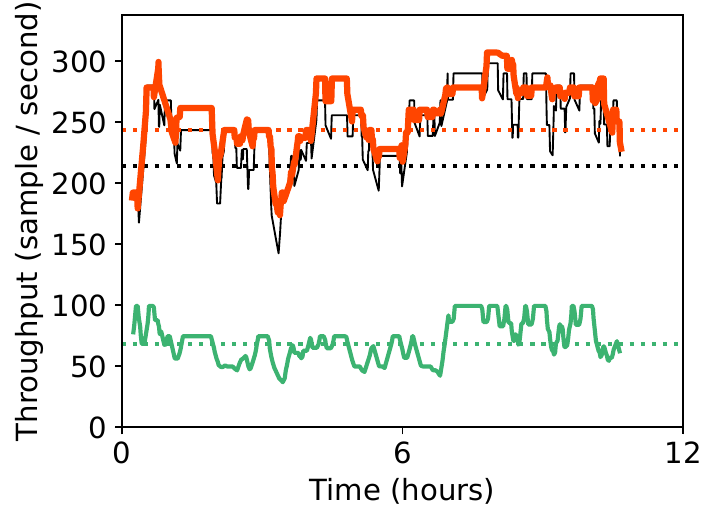}
        \caption{BERT-Large}
        \label{fig:tput_gcp_bert}
    \end{subfigure}%
    \begin{subfigure}[t]{0.25\textwidth}
        \includegraphics[width=\textwidth]{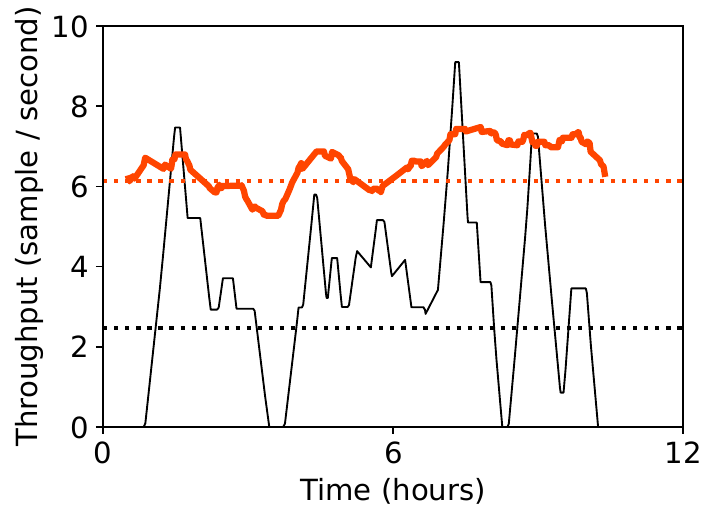}
        \caption{GPT-3 2.7b}
        \label{fig:tput_gcp_gpt327b}
    \end{subfigure}%
    \begin{subfigure}[t]{0.25\textwidth}
        \includegraphics[width=\textwidth]{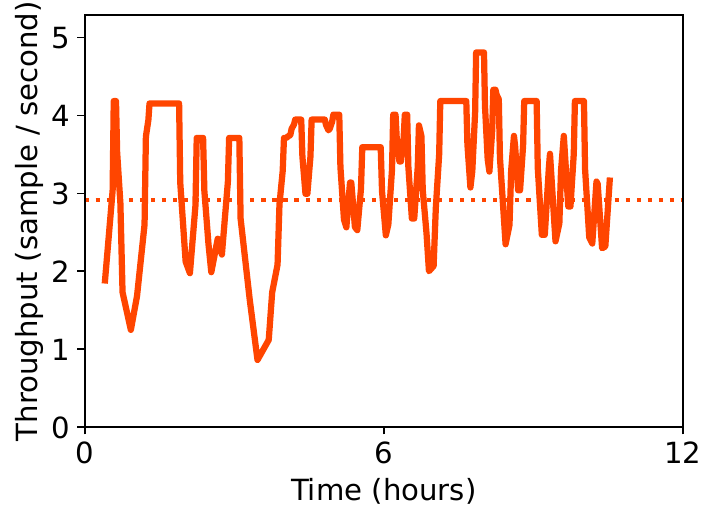}
        \caption{GPT-3 6.7b}
        \label{fig:tput_gcp_gpt367b}
    \end{subfigure}
    \caption{Throughput changes in spot instances environment, EC2 \texttt{P3} instances (top) and GCP \texttt{a2-highgpu-1g} instances (bottom), with various models. Note the different Y-axes scales for different models. Horizontal dotted lines represent average throughput.
    Bamboo could not run any GPT-3 models, while Varuna failed for GPT-3 6.7b.}
    \label{fig:tput_time}
\end{figure*}
 
\subsection{Throughput Under Controlled Failures}
\label{sec:tp_failure_freqs}

We first evaluate the average throughput of Bamboo, Varuna, and \name on various failure scenarios.
We set the frequency of failures from once every 6 hours (low rate) to once every 10 minutes (high rate) to cover a wide spectrum of environments \cite{2023bamboo,2022varuna,2022spotlake,2021nodefailurehpc,2010studyfailurehpc}.
We monotonically reduce the number of available nodes without node recovery and measure average throughput until less than half of the nodes (15 nodes) remain.

Table~\ref{tab:tput_freq} shows the average throughput for different frequencies of failures.
\name outperforms or matches other frameworks for every model in every scenario.
Due to static overhead coming from redundant computation, Bamboo's throughput, while stable over different failure frequencies, is consistently low.
Also, because they need to hold a large portion of GPU memory for an additional copy of model states for redundant computation and activations, Bamboo cannot train large models due to out-of-memory (OOM) errors.

Bamboo's gap from Varuna was surprising.
We believe that it is due to differences in our evaluation environments.
We use 200Gbps high-performance networking and ample NVMe storage, while the original evaluation used Amazon EC2 \texttt{p3.2xlarge} and \texttt{p3.8xlarge} instances with up to 10Gbps network and S3 object storage. 
High storage throughput in our setup significantly sped up Varuna.
Furthermore, we use different batch configurations for Bamboo and Varuna so that each can run at maximum resource utilization; in contrast, Bamboo's evaluation used the same configuration for both, which throttled Varuna's potential throughput.

Overall, Varuna performs comparably to us when either the model is small or failures happen infrequently.
For larger models and/or more frequent failures, the higher overhead of loading and saving checkpoints drastically decrease its throughput ($29.6 \times$ for GPT3-6.7B).

\subsection{Throughput in Spot Instance Traces}
\label{sec:trace_driven_tp}

Next, we borrow real traces of node availability changes of spot instances from the Bamboo repository~\cite{2023bamboogithub} and use their tools to replay the trace for 12 hours \cite{2023bamboo}.
Events in the trace had been gathered from Amazon EC2 \texttt{P3} spot instances (\texttt{p3.2xlarge} and \texttt{p3.8xlarge}) and Google Cloud Platform (GCP) \texttt{a2-highgpu-1g} spot instances.
Node preemption events happen every 7.7 minutes and 10.3 minutes, on average, for EC2 and GCP spot instances, respectively.
Unlike experiments earlier where the number of available nodes monotonically decreases (\S\ref{sec:tp_failure_freqs}), these traces include node addition events too.
The actual experiments took place in our cluster where we simulated the availability events.

Figure~\ref{fig:tput_time} represents throughput changes for some models. See Appendix~\ref{sec:apdx_spot_instances} for results from other models; omitted models are similar to the BERT-Large model.
Note that each data point in lines is an average throughput of a short time window for visibility.
As such, it may not represent 0 throughput, which happens during reconfiguration or full restart.
BERT-Large, the smallest model, has the least amount of checkpointing overhead; as such, the performance of Varuna is similar to \name.
However, as models become larger, even in our high-performance storage setup, Varuna takes increasingly longer to store and load checkpoints.
Also, it starts suffering from fallbacks because it fails to finish checkpointing within the preemption grace period, decreasing its throughput more drastically.
For example, for GPT-3 with 6.7 billion parameters, Varuna hung over the entire time and could not make training progress.
Frequent changes in node availability trigger Varuna's full reconfiguration more frequently, wasting resources for saving and loading checkpoints which decreases its throughput.

Bamboo cannot run any model larger than GPT-2.

\pagebreak
\subsection{Ablation Study}
\label{sec:ablation}

\subsubsection{Overhead of \name Planning.}


\begin{table}
    \small
    \caption{\name planning latency (in seconds) with various numbers of layers and nodes. BERT-Large, GPT-2, and GPT-3 Medium have 24 layers, while GPT-3 2.7b and 6.7b have 32 layers.}
    \label{tab:planning_latency}
    \begin{tabular}{cc|R{1cm}R{1cm}R{1cm}R{1cm}}
    \toprule
    \multirow{2}{*}{\# Nodes} & \multirow{2}{*}{\shortstack{\# GPUs\\Per Node}} & \multicolumn{4}{c}{\# Layers}      \\
                          &                                                                            & 24    & 32    & 64     & 96        \\ \midrule
\multirow{3}{*}{8}        & 1                                                                          & 0.28  & 0.71  & 9.65  & 68.50     \\
                          & 4                                                                          & 0.41  & 1.15  & 11.58  & 74.56     \\
                          & 8                                                                          & 0.54  & 1.50  & 20.98  & 109.76    \\ \midrule
\multirow{3}{*}{16}       & 1                                                                          & 3.37  & 7.45  & 66.35  & 540.36    \\
                          & 4                                                                          & 4.56  & 10.41 & 108.10 & 649.67    \\
                          & 8                                                                          & 4.90  & 11.78 & 176.04 & 1,213.63  \\ \midrule
\multirow{3}{*}{24}       & 1                                                                          & 11.35  & 30.11 & 262.47 & 1,477.54  \\
                          & 4                                                                          & 14.78 & 45.80 & 472.53 & 2,153.84  \\
                          & 8                                                                          & 15.59 & 49.25 & 520.08 & 3,297.92 \\
    \bottomrule
    \end{tabular}
\end{table}

We run the planning algorithm (\S\ref{sec:planning}) to create a single pipeline template using various model specifications (number of layers) and node specifications (number of nodes and GPUs per node) to see its scalability.
Table~\ref{tab:planning_latency} shows the planning algorithm latency with various numbers of layers and nodes.
Considering the estimated end-to-end training time of large models using hundreds of GPUs is $\sim$100 days~\cite{2021megatron,2020gpt3}, the planning overhead is marginal ($<0.1\%$).
Even if more GPUs are used for training (\ie, thousands of GPUs), the number of nodes for each pipeline template does not increase significantly.
This is because \name simply instantiates more of the smaller pipelines and utilizes data parallelism.
Also, once a pipeline template is generated, the creation of subsequent templates can drastically be accelerated, adding negligible time, thanks to the usage of memoization and intermediate caches (\S\ref{sec:dc_algorithm}).


\subsubsection{Throughput Breakdown.}

\begin{figure}
    \begin{subfigure}{\linewidth}
        \centering
        \includegraphics[width=\textwidth]{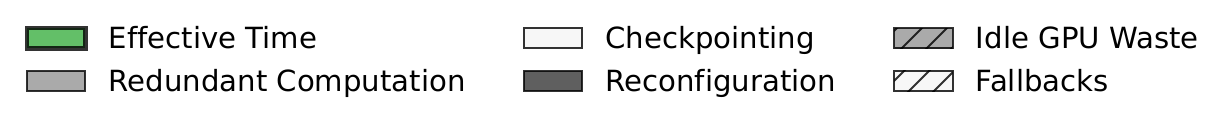}
    \end{subfigure}
    \begin{subfigure}{\linewidth}
        \includegraphics[width=\linewidth]{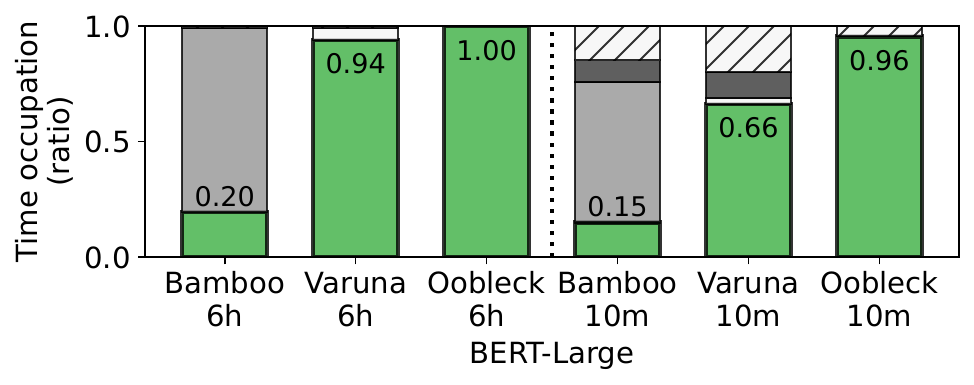}
    \end{subfigure}
    \begin{subfigure}{\linewidth}
        \includegraphics[width=\linewidth]{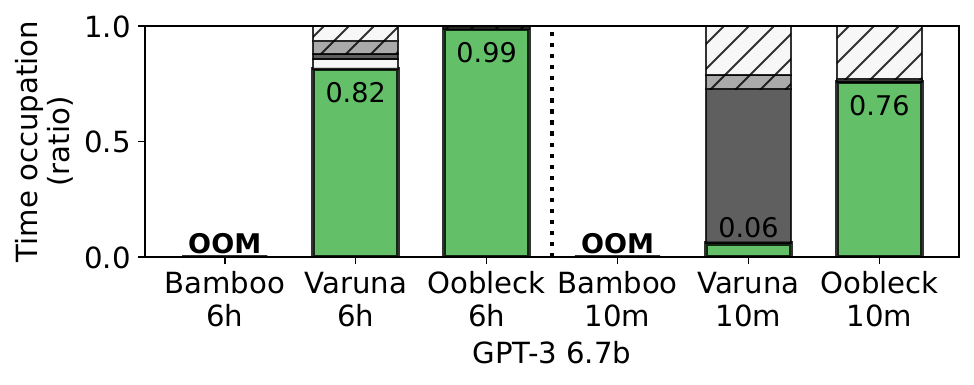}
    \end{subfigure}
    \caption{Time occupation breakdown of Bamboo, Varuna, and \name running BERT-Large and GPT-3 6.7b model.}
    \label{fig:overhead_breakdown}
\end{figure}

Figure~\ref{fig:overhead_breakdown} shows the impact of each overhead on training throughput.
Varuna's overheads include restarting overhead (reinitialization and loading a checkpoint), saving checkpoints, and throughput loss due to idle GPUs and fallbacks.
Bamboo has significant overhead from redundant computation and reconfiguration overhead for data copy.
Redundant computation in Bamboo is shown to have more than 50\% overhead because it includes several indirect components that cannot clearly be separated -- for example, pipeline bubble due to an increased number of pipeline stages to store redundant model states and imbalanced pipeline stages for balancing memory.
\name only has a small copying overhead for missing layers after pipeline reinstantiation.

All the frameworks experience fallback overhead, losing some training progress due to failures happening in the middle of iteration.
It is more severe in Varuna because it has to fall back to the last checkpoint, while Bamboo and \name lose at most one iteration.
Varuna suffers significantly from much such waste occupying up to 95\% of wall clock time, while \name can achieve effective throughput of at least 75\% of the no-failure scenario.


\begin{table}
    \footnotesize
    \caption{Throughput (samples/s) for Varuna, Varuna with no checkpointing overhead, and \name running BERT-Large and GPT-3 6.7b.}
    \label{tab:tp_failure_freqs_nockpt}
    \begin{tabular}{l|rrr|rrr}
    \toprule
                     & \multicolumn{3}{c|}{BERT-Large} & \multicolumn{3}{c}{GPT-3 6.7b}    \\ \midrule
    Failure Frequency  & 6h   & 1h   & 10m  & 6h   & 1h   & 10m  \\ \midrule
    Varuna           & 259.57 & 245.39 & 168.15 & 4.02 & 2.91 & 0.12 \\
    Varuna (no ckpt) & 275.27 & 270.81 & 245.78 & 4.57 & 4.00 & 0.36 \\
    Oobleck          & 287.10 & 286.28 & 282.11 & 4.33 & 4.22 & 3.55 \\
    \bottomrule
    \end{tabular}
\end{table}

\subsubsection{Impact of Checkpointing Overhead.}
\label{sec:no_checkpoint_overhead}
Overhead of fault tolerance in Varuna mostly comes from serialized checkpointing and full restart.
CheckFreq~\cite{2021checkfreq} recently introduced checkpointing optimization by pipelining checkpointing with computation, and it can improve the throughput of checkpoint-based training.
Here, we go further by \emph{completely} removing the overhead of checkpointing and analyzing the impact of failures only.
Because of lower checkpointing overhead, we also increase the frequency of checkpointing from every 10 iterations to every 2 iterations.
We define full restart overhead as framework initialization plus loading the last checkpoint overhead.
While checkpointing overhead during training can be hidden by overlapping it with computation, the overhead of loading a checkpoint cannot be overlapped with computation; this is because computation cannot begin until the entire checkpoint is loaded.

Table~\ref{tab:tp_failure_freqs_nockpt} compares Varuna, Varuna with no checkpointing overhead, and \name running the BERT-Large model and GPT-3 6.7b model.
Although Varuna could increase its throughput, it still suffers from up to 90\% overhead for the higher frequency of failures.

\section{Related Works}
\label{sec:related}

\parabf{Elastic training.} 
Horovod Elastic~\cite{horovodelastic} and TorchElastic~\cite{torchelastic} restart training upon failure and recovery.
CoDDL~\cite{2021coddl} balances resource efficiency and short job priority in elastic resource sharing problems.
Aryl~\cite{2022aryl} enables elastic resource sharing between inference and training workloads.
Pollux~\cite{2021pollux} considers both resource utilization and statistical efficiency of training jobs when adaptively allocating resources.
These works are all limited to elastic resource sharing for data-parallel training of small models that fit within a single GPU.

\parabf{Distributed training with spot instances.} 
Varuna~\cite{2022varuna} uses hybrid parallelism for distributed
training with cheaper spot cloud instances.
It reconfigures training when one or more failures happen. 
Bamboo~\cite{2023bamboo} introduces redundant computation (RC) in pipeline parallelism to provide resilience in the presence of frequent preemptions of training with spot instances.
\name matches or significantly outperforms them for a wide range of model sizes and failure frequencies. 

\parabf{Large model training.}
Numerous proposals in recent years have attempted to optimize large model training through diverse mechanisms \cite{2021megatron, 2019pipedream, 2019gpipe, 2021dapple, 2021piper, 2020hetpipe, 2021chimera, 2022alpa, 2020zero, 2021zerooffload, 2021zeroinfinity, 2022turingnlg}.
However, they do not provide fault tolerance out-of-the-box and are orthogonal to \name.

\section{Conclusion}

In this paper, we introduced \name, a resilient distributed Large model training framework with guaranteed fault tolerance.
\name co-designs planning and execution for fast failure recovery and high throughput by introducing pipeline templates that are carefully designed during planning and reused during training execution.
It achieves efficient failure recovery by reinstantiating pipeline(s) from the pipeline templates and copying missing model states from pipeline replicas without requiring a full restart from checkpoints.
\name outperforms state-of-the-art fault-tolerant distributed training solutions Bamboo and Varuna by up to $29.6\times$.

\section*{Acknowledgements}
We would like to thank the SOSP reviewers, our shepherd Keval Vora, and SymbioticLab members for their insightful feedback. 
This work is in part supported by NSF grants CNS-1909067, CNS-2104243, and CNS-2106184 as well as gifts from VMWare, Google, and Meta.

\label{EndOfPaper}

{
  \bibliographystyle{ACM-Reference-Format}
  \bibliography{ref}


\begin{thebibliography}{56}


\ifx \showCODEN    \undefined \def \showCODEN     #1{\unskip}     \fi
\ifx \showDOI      \undefined \def \showDOI       #1{#1}\fi
\ifx \showISBNx    \undefined \def \showISBNx     #1{\unskip}     \fi
\ifx \showISBNxiii \undefined \def \showISBNxiii  #1{\unskip}     \fi
\ifx \showISSN     \undefined \def \showISSN      #1{\unskip}     \fi
\ifx \showLCCN     \undefined \def \showLCCN      #1{\unskip}     \fi
\ifx \shownote     \undefined \def \shownote      #1{#1}          \fi
\ifx \showarticletitle \undefined \def \showarticletitle #1{#1}   \fi
\ifx \showURL      \undefined \def \showURL       {\relax}        \fi
\providecommand\bibfield[2]{#2}
\providecommand\bibinfo[2]{#2}
\providecommand\natexlab[1]{#1}
\providecommand\showeprint[2][]{arXiv:#2}

\bibitem[Athlur et~al\mbox{.}(2022)]%
        {2022varuna}
\bibfield{author}{\bibinfo{person}{Sanjith Athlur}, \bibinfo{person}{Nitika
  Saran}, \bibinfo{person}{Muthian Sivathanu}, \bibinfo{person}{Ramachandran
  Ramjee}, {and} \bibinfo{person}{Nipun Kwatra}.}
  \bibinfo{year}{2022}\natexlab{}.
\newblock \showarticletitle{Varuna: Scalable, Low-Cost Training of Massive Deep
  Learning Models}. In \bibinfo{booktitle}{\emph{Proceedings of the European
  Conference on Computer Systems (EuroSys)}}.
\newblock
\urldef\tempurl%
\url{https://doi.org/10.1145/3492321.3519584}
\showDOI{\tempurl}


\bibitem[Beaumont(2022)]%
        {2022laion}
\bibfield{author}{\bibinfo{person}{Romain Beaumont}.}
  \bibinfo{year}{2022}\natexlab{}.
\newblock \bibinfo{title}{Large Scale OpenCLIP: L/14, H/14 AND G/14 Trained on
  LAION-2B}.
\newblock
\newblock
\urldef\tempurl%
\url{https://laion.ai/blog/large-openclip/}
\showURL{%
\tempurl}


\bibitem[Bekman(2022)]%
        {2022bloom}
\bibfield{author}{\bibinfo{person}{Stas Bekman}.}
  \bibinfo{year}{2022}\natexlab{}.
\newblock \bibinfo{title}{The Technology Behind Bloom Training}.
\newblock
\newblock
\urldef\tempurl%
\url{https://huggingface.co/blog/bloom-megatron-deepspeed}
\showURL{%
\tempurl}


\bibitem[Bellman(1952)]%
        {coin_change}
\bibfield{author}{\bibinfo{person}{Richard Bellman}.}
  \bibinfo{year}{1952}\natexlab{}.
\newblock \showarticletitle{On the Theory of Dynamic Programming}.
\newblock \bibinfo{journal}{\emph{Proceedings of the National Academy of
  Sciences (PNAS)}}.
\newblock
\urldef\tempurl%
\url{https://doi.org/10.1073/pnas.38.8.716}
\showDOI{\tempurl}


\bibitem[Brown et~al\mbox{.}(2020)]%
        {2020gpt3}
\bibfield{author}{\bibinfo{person}{Tom Brown}, \bibinfo{person}{Benjamin Mann},
  \bibinfo{person}{Nick Ryder}, \bibinfo{person}{Melanie Subbiah},
  \bibinfo{person}{Jared~D Kaplan}, \bibinfo{person}{Prafulla Dhariwal},
  \bibinfo{person}{Arvind Neelakantan}, \bibinfo{person}{Pranav Shyam},
  \bibinfo{person}{Girish Sastry}, \bibinfo{person}{Amanda Askell},
  \bibinfo{person}{Sandhini Agarwal}, \bibinfo{person}{Ariel Herbert-Voss},
  \bibinfo{person}{Gretchen Krueger}, \bibinfo{person}{Tom Henighan},
  \bibinfo{person}{Rewon Child}, \bibinfo{person}{Aditya Ramesh},
  \bibinfo{person}{Daniel Ziegler}, \bibinfo{person}{Jeffrey Wu},
  \bibinfo{person}{Clemens Winter}, \bibinfo{person}{Chris Hesse},
  \bibinfo{person}{Mark Chen}, \bibinfo{person}{Eric Sigler},
  \bibinfo{person}{Mateusz Litwin}, \bibinfo{person}{Scott Gray},
  \bibinfo{person}{Benjamin Chess}, \bibinfo{person}{Jack Clark},
  \bibinfo{person}{Christopher Berner}, \bibinfo{person}{Sam McCandlish},
  \bibinfo{person}{Alec Radford}, \bibinfo{person}{Ilya Sutskever}, {and}
  \bibinfo{person}{Dario Amodei}.} \bibinfo{year}{2020}\natexlab{}.
\newblock \showarticletitle{Language Models are Few-Shot Learners}. In
  \bibinfo{booktitle}{\emph{Advances in Neural Information Processing Systems
  (NeurIPS)}}.
\newblock
\urldef\tempurl%
\url{https://proceedings.neurips.cc/paper_files/paper/2020/hash/1457c0d6bfcb4967418bfb8ac142f64a-Abstract.html}
\showURL{%
\tempurl}


\bibitem[Bynum et~al\mbox{.}(2021)]%
        {2021pyomo}
\bibfield{author}{\bibinfo{person}{Michael~L. Bynum},
  \bibinfo{person}{Gabriel~A. Hackebeil}, \bibinfo{person}{William~E. Hart},
  \bibinfo{person}{Carl~D. Laird}, \bibinfo{person}{Bethany~L. Nicholson},
  \bibinfo{person}{John~D. Siirola}, \bibinfo{person}{Jean-Paul Watson}, {and}
  \bibinfo{person}{David~L. Woodruff}.} \bibinfo{year}{2021}\natexlab{}.
\newblock \bibinfo{booktitle}{\emph{Pyomo--Optimization Modeling in Python}}.
\newblock \bibinfo{publisher}{Springer Science \& Business Media}.
\newblock
\showISBNx{9783030689278}
\urldef\tempurl%
\url{https://doi.org/10.1007/978-3-030-68928-5}
\showDOI{\tempurl}


\bibitem[Chen et~al\mbox{.}(2016)]%
        {2016activationrecomputation}
\bibfield{author}{\bibinfo{person}{Tianqi Chen}, \bibinfo{person}{Bing Xu},
  \bibinfo{person}{Chiyuan Zhang}, {and} \bibinfo{person}{Carlos Guestrin}.}
  \bibinfo{year}{2016}\natexlab{}.
\newblock \bibinfo{title}{Training Deep Nets with Sublinear Memory Cost}.
\newblock
\newblock
\showeprint[arxiv]{1604.06174}~[cs.LG]


\bibitem[Das et~al\mbox{.}(2021)]%
        {2021nodefailurehpc}
\bibfield{author}{\bibinfo{person}{Anwesha Das}, \bibinfo{person}{Frank
  Mueller}, {and} \bibinfo{person}{Barry Rountree}.}
  \bibinfo{year}{2021}\natexlab{}.
\newblock \showarticletitle{Systemic Assessment of Node Failures in HPC
  Production Platforms}. In \bibinfo{booktitle}{\emph{International Parallel
  and Distributed Processing Symposium (IPDPS)}}.
\newblock
\urldef\tempurl%
\url{https://doi.org/10.1109/IPDPS49936.2021.00035}
\showDOI{\tempurl}


\bibitem[Devlin et~al\mbox{.}(2019)]%
        {2019bert}
\bibfield{author}{\bibinfo{person}{Jacob Devlin}, \bibinfo{person}{Ming-Wei
  Chang}, \bibinfo{person}{Kenton Lee}, {and} \bibinfo{person}{Kristina
  Toutanova}.} \bibinfo{year}{2019}\natexlab{}.
\newblock \showarticletitle{BERT: Pre-training of Deep Bidirectional
  Transformers for Language Understanding}. In
  \bibinfo{booktitle}{\emph{Conference of the North American Chapter of the
  Association for Computational Linguistics (NAACL)}}.
\newblock
\urldef\tempurl%
\url{https://doi.org/10.18653/v1/N19-1423}
\showDOI{\tempurl}


\bibitem[Eisenman et~al\mbox{.}(2022)]%
        {2022checknrun}
\bibfield{author}{\bibinfo{person}{Assaf Eisenman},
  \bibinfo{person}{Kiran~Kumar Matam}, \bibinfo{person}{Steven Ingram},
  \bibinfo{person}{Dheevatsa Mudigere}, \bibinfo{person}{Raghuraman
  Krishnamoorthi}, \bibinfo{person}{Krishnakumar Nair}, \bibinfo{person}{Misha
  Smelyanskiy}, {and} \bibinfo{person}{Murali Annavaram}.}
  \bibinfo{year}{2022}\natexlab{}.
\newblock \showarticletitle{Check-N-Run: a Checkpointing System for Training
  Deep Learning Recommendation Models}. In \bibinfo{booktitle}{\emph{USENIX
  Symposium on Networked Systems Design and Implementation (NSDI)}}.
\newblock
\urldef\tempurl%
\url{https://www.usenix.org/conference/nsdi22/presentation/eisenman}
\showURL{%
\tempurl}


\bibitem[Fan et~al\mbox{.}(2021)]%
        {2021dapple}
\bibfield{author}{\bibinfo{person}{Shiqing Fan}, \bibinfo{person}{Yi Rong},
  \bibinfo{person}{Chen Meng}, \bibinfo{person}{Zongyan Cao},
  \bibinfo{person}{Siyu Wang}, \bibinfo{person}{Zhen Zheng},
  \bibinfo{person}{Chuan Wu}, \bibinfo{person}{Guoping Long},
  \bibinfo{person}{Jun Yang}, \bibinfo{person}{Lixue Xia},
  \bibinfo{person}{Lansong Diao}, \bibinfo{person}{Xiaoyong Liu}, {and}
  \bibinfo{person}{Wei Lin}.} \bibinfo{year}{2021}\natexlab{}.
\newblock \showarticletitle{DAPPLE: A Pipelined Data Parallel Approach for
  Training Large Models}. In \bibinfo{booktitle}{\emph{Symposium on Principles
  and Practice of Parallel Programming (PPoPP)}}.
\newblock
\urldef\tempurl%
\url{https://doi.org/10.1145/3437801.3441593}
\showDOI{\tempurl}


\bibitem[Gupta et~al\mbox{.}(2017)]%
        {2017failuresinlargescale}
\bibfield{author}{\bibinfo{person}{Saurabh Gupta}, \bibinfo{person}{Tirthak
  Patel}, \bibinfo{person}{Christian Engelmann}, {and} \bibinfo{person}{Devesh
  Tiwari}.} \bibinfo{year}{2017}\natexlab{}.
\newblock \showarticletitle{Failures in Large Scale Systems: Long-Term
  Measurement, Analysis, and Implications}. In
  \bibinfo{booktitle}{\emph{International Conference for High Performance
  Computing, Networking, Storage, and Analysis (SC)}}.
\newblock
\urldef\tempurl%
\url{https://doi.org/10.1145/3126908.3126937}
\showDOI{\tempurl}


\bibitem[Hart et~al\mbox{.}(2011)]%
        {2011pyomo}
\bibfield{author}{\bibinfo{person}{William~E Hart}, \bibinfo{person}{Jean-Paul
  Watson}, {and} \bibinfo{person}{David~L Woodruff}.}
  \bibinfo{year}{2011}\natexlab{}.
\newblock \showarticletitle{Pyomo: Modeling and Solving Mathematical Programs
  in Python}.
\newblock \bibinfo{journal}{\emph{Mathematical Programming Computation (MPC)}}
  (\bibinfo{year}{2011}).
\newblock
\urldef\tempurl%
\url{https://doi.org/10.1007/s12532-011-0026-8}
\showDOI{\tempurl}


\bibitem[Horovod(2019)]%
        {horovodelastic}
\bibfield{author}{\bibinfo{person}{Horovod}.} \bibinfo{year}{2019}\natexlab{}.
\newblock \bibinfo{title}{Elastic Horovod}.
\newblock
\newblock
\urldef\tempurl%
\url{https://horovod.readthedocs.io/en/stable/elastic_include.html}
\showURL{%
\tempurl}


\bibitem[Huang et~al\mbox{.}(2019)]%
        {2019gpipe}
\bibfield{author}{\bibinfo{person}{Yanping Huang}, \bibinfo{person}{Youlong
  Cheng}, \bibinfo{person}{Ankur Bapna}, \bibinfo{person}{Orhan Firat},
  \bibinfo{person}{Dehao Chen}, \bibinfo{person}{Mia Chen},
  \bibinfo{person}{HyoukJoong Lee}, \bibinfo{person}{Jiquan Ngiam},
  \bibinfo{person}{Quoc~V Le}, \bibinfo{person}{Yonghui Wu}, {and}
  \bibinfo{person}{zhifeng Chen}.} \bibinfo{year}{2019}\natexlab{}.
\newblock \showarticletitle{GPipe: Efficient Training of Giant Neural Networks
  using Pipeline Parallelism}. In \bibinfo{booktitle}{\emph{Advances in Neural
  Information Processing Systems (NeurIPS)}}.
\newblock
\urldef\tempurl%
\url{https://proceedings.neurips.cc/paper_files/paper/2019/hash/093f65e080a295f8076b1c5722a46aa2-Abstract.html}
\showURL{%
\tempurl}


\bibitem[Hwang et~al\mbox{.}(2021)]%
        {2021coddl}
\bibfield{author}{\bibinfo{person}{Changho Hwang}, \bibinfo{person}{Taehyun
  Kim}, \bibinfo{person}{Sunghyun Kim}, \bibinfo{person}{Jinwoo Shin}, {and}
  \bibinfo{person}{KyoungSoo Park}.} \bibinfo{year}{2021}\natexlab{}.
\newblock \showarticletitle{Elastic Resource Sharing for Distributed Deep
  Learning}. In \bibinfo{booktitle}{\emph{USENIX Symposium on Networked Systems
  Design and Implementation (NSDI)}}.
\newblock
\urldef\tempurl%
\url{https://www.usenix.org/conference/nsdi21/presentation/hwang}
\showURL{%
\tempurl}


\bibitem[Jeon et~al\mbox{.}(2019)]%
        {2019philly}
\bibfield{author}{\bibinfo{person}{Myeongjae Jeon}, \bibinfo{person}{Shivaram
  Venkataraman}, \bibinfo{person}{Amar Phanishayee}, \bibinfo{person}{Junjie
  Qian}, \bibinfo{person}{Wencong Xiao}, {and} \bibinfo{person}{Fan Yang}.}
  \bibinfo{year}{2019}\natexlab{}.
\newblock \showarticletitle{Analysis of Large-Scale Multi-Tenant GPU Clusters
  for DNN Training Workloads}. In \bibinfo{booktitle}{\emph{USENIX Annual
  Technical Conference (ATC)}}.
\newblock
\urldef\tempurl%
\url{https://www.usenix.org/conference/atc19/presentation/jeon}
\showURL{%
\tempurl}


\bibitem[Jia et~al\mbox{.}(2022)]%
        {2022whale}
\bibfield{author}{\bibinfo{person}{Xianyan Jia}, \bibinfo{person}{Le Jiang},
  \bibinfo{person}{Ang Wang}, \bibinfo{person}{Wencong Xiao},
  \bibinfo{person}{Ziji Shi}, \bibinfo{person}{Jie Zhang},
  \bibinfo{person}{Xinyuan Li}, \bibinfo{person}{Langshi Chen},
  \bibinfo{person}{Yong Li}, \bibinfo{person}{Zhen Zheng},
  \bibinfo{person}{Xiaoyong Liu}, {and} \bibinfo{person}{Wei Lin}.}
  \bibinfo{year}{2022}\natexlab{}.
\newblock \showarticletitle{Whale: Efficient Giant Model Training over
  Heterogeneous GPUs}. In \bibinfo{booktitle}{\emph{USENIX Annual Technical
  Conference (ATC)}}.
\newblock
\urldef\tempurl%
\url{https://www.usenix.org/conference/atc22/presentation/jia-xianyan}
\showURL{%
\tempurl}


\bibitem[Lai et~al\mbox{.}(2023)]%
        {2023merak}
\bibfield{author}{\bibinfo{person}{Zhiquan Lai}, \bibinfo{person}{Shengwei Li},
  \bibinfo{person}{Xudong Tang}, \bibinfo{person}{Keshi Ge},
  \bibinfo{person}{Weijie Liu}, \bibinfo{person}{Yabo Duan},
  \bibinfo{person}{Linbo Qiao}, {and} \bibinfo{person}{Dongsheng Li}.}
  \bibinfo{year}{2023}\natexlab{}.
\newblock \showarticletitle{Merak: An Efficient Distributed DNN Training
  Framework With Automated 3D Parallelism for Giant Foundation Models}.
\newblock \bibinfo{journal}{\emph{Transactions on Parallel and Distributed
  Systems (TPDS)}} (\bibinfo{year}{2023}).
\newblock
\urldef\tempurl%
\url{https://doi.org/10.1109/TPDS.2023.3247001}
\showDOI{\tempurl}


\bibitem[Lee et~al\mbox{.}(2022)]%
        {2022spotlake}
\bibfield{author}{\bibinfo{person}{Sungjae Lee}, \bibinfo{person}{Jaeil Hwang},
  {and} \bibinfo{person}{Kyungyong Lee}.} \bibinfo{year}{2022}\natexlab{}.
\newblock \showarticletitle{SpotLake: Diverse Spot Instance Dataset Archive
  Service}. In \bibinfo{booktitle}{\emph{International Symposium on Workload
  Characterization (IISWC)}}.
\newblock
\urldef\tempurl%
\url{https://doi.org/10.1109/IISWC55918.2022.00029}
\showDOI{\tempurl}


\bibitem[Li et~al\mbox{.}(2022)]%
        {2022aryl}
\bibfield{author}{\bibinfo{person}{Jiamin Li}, \bibinfo{person}{Hong Xu},
  \bibinfo{person}{Yibo Zhu}, \bibinfo{person}{Zherui Liu},
  \bibinfo{person}{Chuanxiong Guo}, {and} \bibinfo{person}{Cong Wang}.}
  \bibinfo{year}{2022}\natexlab{}.
\newblock \bibinfo{title}{Aryl: An Elastic Cluster Scheduler for Deep
  Learning}.
\newblock
\newblock
\showeprint[arxiv]{2202.07896}~[cs.DC]


\bibitem[Li and Hoefler(2021)]%
        {2021chimera}
\bibfield{author}{\bibinfo{person}{Shigang Li} {and} \bibinfo{person}{Torsten
  Hoefler}.} \bibinfo{year}{2021}\natexlab{}.
\newblock \showarticletitle{Chimera: Efficiently Training Large-Scale Neural
  Networks with Bidirectional Pipelines}. In
  \bibinfo{booktitle}{\emph{International Conference for High Performance
  Computing, Networking, Storage and Analysis (SC)}}.
\newblock
\urldef\tempurl%
\url{https://doi.org/10.1145/3458817.3476145}
\showDOI{\tempurl}


\bibitem[Li et~al\mbox{.}(2020)]%
        {2020pytorch}
\bibfield{author}{\bibinfo{person}{Shen Li}, \bibinfo{person}{Yanli Zhao},
  \bibinfo{person}{Rohan Varma}, \bibinfo{person}{Omkar Salpekar},
  \bibinfo{person}{Pieter Noordhuis}, \bibinfo{person}{Teng Li},
  \bibinfo{person}{Adam Paszke}, \bibinfo{person}{Jeff Smith},
  \bibinfo{person}{Brian Vaughan}, \bibinfo{person}{Pritam Damania}, {and}
  \bibinfo{person}{Soumith Chintala}.} \bibinfo{year}{2020}\natexlab{}.
\newblock \showarticletitle{PyTorch Distributed: Experiences on Accelerating
  Data Parallel Training}.
\newblock \bibinfo{journal}{\emph{Proceedings of VLDB Endowment (VLDB)}}
  (\bibinfo{year}{2020}).
\newblock
\urldef\tempurl%
\url{https://doi.org/10.14778/3415478.3415530}
\showDOI{\tempurl}


\bibitem[Lian et~al\mbox{.}(2022)]%
        {2022persia}
\bibfield{author}{\bibinfo{person}{Xiangru Lian}, \bibinfo{person}{Binhang
  Yuan}, \bibinfo{person}{Xuefeng Zhu}, \bibinfo{person}{Yulong Wang},
  \bibinfo{person}{Yongjun He}, \bibinfo{person}{Honghuan Wu},
  \bibinfo{person}{Lei Sun}, \bibinfo{person}{Haodong Lyu},
  \bibinfo{person}{Chengjun Liu}, \bibinfo{person}{Xing Dong},
  \bibinfo{person}{Yiqiao Liao}, \bibinfo{person}{Mingnan Luo},
  \bibinfo{person}{Congfei Zhang}, \bibinfo{person}{Jingru Xie},
  \bibinfo{person}{Haonan Li}, \bibinfo{person}{Lei Chen},
  \bibinfo{person}{Renjie Huang}, \bibinfo{person}{Jianying Lin},
  \bibinfo{person}{Chengchun Shu}, \bibinfo{person}{Xuezhong Qiu},
  \bibinfo{person}{Zhishan Liu}, \bibinfo{person}{Dongying Kong},
  \bibinfo{person}{Lei Yuan}, \bibinfo{person}{Hai Yu}, \bibinfo{person}{Sen
  Yang}, \bibinfo{person}{Ce Zhang}, {and} \bibinfo{person}{Ji Liu}.}
  \bibinfo{year}{2022}\natexlab{}.
\newblock \showarticletitle{Persia: An Open, Hybrid System Scaling Deep
  Learning-Based Recommenders up to 100 Trillion Parameters}. In
  \bibinfo{booktitle}{\emph{International Conference on Knowledge Discovery and
  Data Mining (KDD)}}.
\newblock
\urldef\tempurl%
\url{https://doi.org/10.1145/3534678.3539070}
\showDOI{\tempurl}


\bibitem[Merity et~al\mbox{.}(2017)]%
        {2016wikitext}
\bibfield{author}{\bibinfo{person}{Stephen Merity}, \bibinfo{person}{Caiming
  Xiong}, \bibinfo{person}{James Bradbury}, {and} \bibinfo{person}{Richard
  Socher}.} \bibinfo{year}{2017}\natexlab{}.
\newblock \showarticletitle{Pointer Sentinel Mixture Models}. In
  \bibinfo{booktitle}{\emph{International Conference on Learning
  Representations (ICLR)}}.
\newblock
\urldef\tempurl%
\url{https://openreview.net/forum?id=Byj72udxe}
\showURL{%
\tempurl}


\bibitem[Microsoft(2022)]%
        {2022varunagithub}
\bibfield{author}{\bibinfo{person}{Microsoft}.}
  \bibinfo{year}{2022}\natexlab{}.
\newblock \bibinfo{title}{Varuna}.
\newblock
\newblock
\urldef\tempurl%
\url{https://github.com/microsoft/varuna}
\showURL{%
\tempurl}


\bibitem[MinIO(2023)]%
        {2023minio}
\bibfield{author}{\bibinfo{person}{MinIO}.} \bibinfo{year}{2023}\natexlab{}.
\newblock \bibinfo{title}{MinIO: High Performance Object Storage for AI}.
\newblock
\newblock
\urldef\tempurl%
\url{https://github.com/minio/minio}
\showURL{%
\tempurl}


\bibitem[Mohan et~al\mbox{.}(2021)]%
        {2021checkfreq}
\bibfield{author}{\bibinfo{person}{Jayashree Mohan}, \bibinfo{person}{Amar
  Phanishayee}, {and} \bibinfo{person}{Vijay Chidambaram}.}
  \bibinfo{year}{2021}\natexlab{}.
\newblock \showarticletitle{CheckFreq: Frequent, Fine-Grained DNN
  Checkpointing}. In \bibinfo{booktitle}{\emph{USENIX Conference on File and
  Storage Technologies (FAST)}}.
\newblock
\urldef\tempurl%
\url{https://www.usenix.org/conference/fast21/presentation/mohan}
\showURL{%
\tempurl}


\bibitem[Narayanan et~al\mbox{.}(2019)]%
        {2019pipedream}
\bibfield{author}{\bibinfo{person}{Deepak Narayanan}, \bibinfo{person}{Aaron
  Harlap}, \bibinfo{person}{Amar Phanishayee}, \bibinfo{person}{Vivek
  Seshadri}, \bibinfo{person}{Nikhil~R. Devanur}, \bibinfo{person}{Gregory~R.
  Ganger}, \bibinfo{person}{Phillip~B. Gibbons}, {and} \bibinfo{person}{Matei
  Zaharia}.} \bibinfo{year}{2019}\natexlab{}.
\newblock \showarticletitle{PipeDream: Generalized Pipeline Parallelism for DNN
  Training}. In \bibinfo{booktitle}{\emph{Symposium on Operating Systems
  Principles (SOSP)}}.
\newblock
\urldef\tempurl%
\url{https://doi.org/10.1145/3341301.3359646}
\showDOI{\tempurl}


\bibitem[Narayanan et~al\mbox{.}(2021)]%
        {2021megatron}
\bibfield{author}{\bibinfo{person}{Deepak Narayanan}, \bibinfo{person}{Mohammad
  Shoeybi}, \bibinfo{person}{Jared Casper}, \bibinfo{person}{Patrick
  LeGresley}, \bibinfo{person}{Mostofa Patwary}, \bibinfo{person}{Vijay
  Korthikanti}, \bibinfo{person}{Dmitri Vainbrand}, \bibinfo{person}{Prethvi
  Kashinkunti}, \bibinfo{person}{Julie Bernauer}, \bibinfo{person}{Bryan
  Catanzaro}, \bibinfo{person}{Amar Phanishayee}, {and} \bibinfo{person}{Matei
  Zaharia}.} \bibinfo{year}{2021}\natexlab{}.
\newblock \showarticletitle{Efficient Large-Scale Language Model Training on
  GPU Clusters Using Megatron-LM}. In \bibinfo{booktitle}{\emph{International
  Conference for High Performance Computing, Networking, Storage, and Analysis
  (SC)}}.
\newblock
\urldef\tempurl%
\url{https://doi.org/10.1145/3458817.3476209}
\showDOI{\tempurl}


\bibitem[Naumov et~al\mbox{.}(2019)]%
        {2019dlrm}
\bibfield{author}{\bibinfo{person}{Maxim Naumov}, \bibinfo{person}{Dheevatsa
  Mudigere}, \bibinfo{person}{Hao-Jun~Michael Shi}, \bibinfo{person}{Jianyu
  Huang}, \bibinfo{person}{Narayanan Sundaraman}, \bibinfo{person}{Jongsoo
  Park}, \bibinfo{person}{Xiaodong Wang}, \bibinfo{person}{Udit Gupta},
  \bibinfo{person}{Carole-Jean Wu}, \bibinfo{person}{Alisson~G. Azzolini},
  \bibinfo{person}{Dmytro Dzhulgakov}, \bibinfo{person}{Andrey Mallevich},
  \bibinfo{person}{Ilia Cherniavskii}, \bibinfo{person}{Yinghai Lu},
  \bibinfo{person}{Raghuraman Krishnamoorthi}, \bibinfo{person}{Ansha Yu},
  \bibinfo{person}{Volodymyr Kondratenko}, \bibinfo{person}{Stephanie Pereira},
  \bibinfo{person}{Xianjie Chen}, \bibinfo{person}{Wenlin Chen},
  \bibinfo{person}{Vijay Rao}, \bibinfo{person}{Bill Jia},
  \bibinfo{person}{Liang Xiong}, {and} \bibinfo{person}{Misha Smelyanskiy}.}
  \bibinfo{year}{2019}\natexlab{}.
\newblock \bibinfo{title}{Deep Learning Recommendation Model for
  Personalization and Recommendation Systems}.
\newblock
\newblock
\showeprint[arxiv]{1906.00091}~[cs.IR]


\bibitem[OpenAI(2023)]%
        {2023gpt4}
\bibfield{author}{\bibinfo{person}{OpenAI}.} \bibinfo{year}{2023}\natexlab{}.
\newblock \bibinfo{title}{GPT-4 Technical Report}.
\newblock
\newblock
\showeprint[arxiv]{2303.08774}~[cs.CL]


\bibitem[Park et~al\mbox{.}(2020)]%
        {2020hetpipe}
\bibfield{author}{\bibinfo{person}{Jay~H. Park}, \bibinfo{person}{Gyeongchan
  Yun}, \bibinfo{person}{Chang~M. Yi}, \bibinfo{person}{Nguyen~T. Nguyen},
  \bibinfo{person}{Seungmin Lee}, \bibinfo{person}{Jaesik Choi},
  \bibinfo{person}{Sam~H. Noh}, {and} \bibinfo{person}{Young ri Choi}.}
  \bibinfo{year}{2020}\natexlab{}.
\newblock \showarticletitle{HetPipe: Enabling Large DNN Training on (Whimpy)
  Heterogeneous GPU Clusters through Integration of Pipelined Model Parallelism
  and Data Parallelism}. In \bibinfo{booktitle}{\emph{USENIX Annual Technical
  Conference (ATC)}}.
\newblock
\urldef\tempurl%
\url{https://www.usenix.org/conference/atc20/presentation/park}
\showURL{%
\tempurl}


\bibitem[PyTorch(2020)]%
        {torchelastic}
\bibfield{author}{\bibinfo{person}{PyTorch}.} \bibinfo{year}{2020}\natexlab{}.
\newblock \bibinfo{title}{Torch Elastic}.
\newblock
\newblock
\urldef\tempurl%
\url{https://pytorch.org/elastic/latest/}
\showURL{%
\tempurl}


\bibitem[Qiao et~al\mbox{.}(2021)]%
        {2021pollux}
\bibfield{author}{\bibinfo{person}{Aurick Qiao}, \bibinfo{person}{Sang~Keun
  Choe}, \bibinfo{person}{Suhas~Jayaram Subramanya}, \bibinfo{person}{Willie
  Neiswanger}, \bibinfo{person}{Qirong Ho}, \bibinfo{person}{Hao Zhang},
  \bibinfo{person}{Gregory~R. Ganger}, {and} \bibinfo{person}{Eric~P. Xing}.}
  \bibinfo{year}{2021}\natexlab{}.
\newblock \showarticletitle{Pollux: Co-adaptive Cluster Scheduling for
  Goodput-Optimized Deep Learning}. In \bibinfo{booktitle}{\emph{USENIX
  Symposium on Operating Systems Design and Implementation (OSDI)}}.
\newblock
\urldef\tempurl%
\url{https://www.usenix.org/conference/osdi21/presentation/qiao}
\showURL{%
\tempurl}


\bibitem[Radford et~al\mbox{.}(2019)]%
        {2019gpt2}
\bibfield{author}{\bibinfo{person}{Alec Radford}, \bibinfo{person}{Jeffrey Wu},
  \bibinfo{person}{Rewon Child}, \bibinfo{person}{David Luan},
  \bibinfo{person}{Dario Amodei}, \bibinfo{person}{Ilya Sutskever},
  {et~al\mbox{.}}} \bibinfo{year}{2019}\natexlab{}.
\newblock \bibinfo{title}{Language Models are Unsupervised Multitask Learners}.
\newblock
\newblock
\urldef\tempurl%
\url{https://d4mucfpksywv.cloudfront.net/better-language-models/language-models.pdf}
\showURL{%
\tempurl}


\bibitem[Rajbhandari et~al\mbox{.}(2020)]%
        {2020zero}
\bibfield{author}{\bibinfo{person}{Samyam Rajbhandari}, \bibinfo{person}{Jeff
  Rasley}, \bibinfo{person}{Olatunji Ruwase}, {and} \bibinfo{person}{Yuxiong
  He}.} \bibinfo{year}{2020}\natexlab{}.
\newblock \showarticletitle{ZeRO: Memory Optimizations Toward Training Trillion
  Parameter Models}. In \bibinfo{booktitle}{\emph{International Conference for
  High Performance Computing, Networking, Storage and Analysis (SC)}}.
\newblock
\urldef\tempurl%
\url{https://doi.org/10.1109/SC41405.2020.00024}
\showDOI{\tempurl}


\bibitem[Rajbhandari et~al\mbox{.}(2021)]%
        {2021zeroinfinity}
\bibfield{author}{\bibinfo{person}{Samyam Rajbhandari},
  \bibinfo{person}{Olatunji Ruwase}, \bibinfo{person}{Jeff Rasley},
  \bibinfo{person}{Shaden Smith}, {and} \bibinfo{person}{Yuxiong He}.}
  \bibinfo{year}{2021}\natexlab{}.
\newblock \showarticletitle{ZeRO-Infinity: Breaking the GPU Memory Wall for
  Extreme Scale Deep Learning}. In \bibinfo{booktitle}{\emph{International
  Conference for High Performance Computing, Networking, Storage, and Analysis
  (SC)}}.
\newblock
\urldef\tempurl%
\url{https://doi.org/10.1145/3458817.3476205}
\showDOI{\tempurl}


\bibitem[Ramírez~Alfonsín(2005)]%
        {2005frobenius}
\bibfield{author}{\bibinfo{person}{Jorge~L. Ramírez~Alfonsín}.}
  \bibinfo{year}{2005}\natexlab{}.
\newblock \bibinfo{booktitle}{\emph{The Diophantine Frobenius Problem}}.
\newblock \bibinfo{publisher}{Oxford University Press}.
\newblock
\showISBNx{9780198568209}
\urldef\tempurl%
\url{https://doi.org/10.1093/acprof:oso/9780198568209.001.0001}
\showDOI{\tempurl}


\bibitem[Rasley et~al\mbox{.}(2020)]%
        {2020deepspeed}
\bibfield{author}{\bibinfo{person}{Jeff Rasley}, \bibinfo{person}{Samyam
  Rajbhandari}, \bibinfo{person}{Olatunji Ruwase}, {and}
  \bibinfo{person}{Yuxiong He}.} \bibinfo{year}{2020}\natexlab{}.
\newblock \showarticletitle{DeepSpeed: System Optimizations Enable Training
  Deep Learning Models with Over 100 Billion Parameters}. In
  \bibinfo{booktitle}{\emph{International Conference on Knowledge Discovery and
  Data Mining (KDD)}}.
\newblock
\urldef\tempurl%
\url{https://doi.org/10.1145/3394486.3406703}
\showDOI{\tempurl}


\bibitem[Reed et~al\mbox{.}(2022)]%
        {2022torchfx}
\bibfield{author}{\bibinfo{person}{James Reed}, \bibinfo{person}{Zachary
  DeVito}, \bibinfo{person}{Horace He}, \bibinfo{person}{Ansley Ussery}, {and}
  \bibinfo{person}{Jason Ansel}.} \bibinfo{year}{2022}\natexlab{}.
\newblock \showarticletitle{torch.fx: Practical Program Capture and
  Transformation for Deep Learning in Python}. In
  \bibinfo{booktitle}{\emph{Proceedings of Machine Learning and Systems
  (MLSys)}}.
\newblock
\urldef\tempurl%
\url{https://proceedings.mlsys.org/paper_files/paper/2022/hash/7c98f9c7ab2df90911da23f9ce72ed6e-Abstract.html}
\showURL{%
\tempurl}


\bibitem[Ren et~al\mbox{.}(2021)]%
        {2021zerooffload}
\bibfield{author}{\bibinfo{person}{Jie Ren}, \bibinfo{person}{Samyam
  Rajbhandari}, \bibinfo{person}{Reza~Yazdani Aminabadi},
  \bibinfo{person}{Olatunji Ruwase}, \bibinfo{person}{Shuangyan Yang},
  \bibinfo{person}{Minjia Zhang}, \bibinfo{person}{Dong Li}, {and}
  \bibinfo{person}{Yuxiong He}.} \bibinfo{year}{2021}\natexlab{}.
\newblock \showarticletitle{ZeRO-Offload: Democratizing Billion-Scale Model
  Training}. In \bibinfo{booktitle}{\emph{USENIX Annual Technical Conference
  (ATC)}}.
\newblock
\urldef\tempurl%
\url{https://www.usenix.org/conference/atc21/presentation/ren-jie}
\showURL{%
\tempurl}


\bibitem[Roberts(1956)]%
        {1956frobenius}
\bibfield{author}{\bibinfo{person}{J.~B. Roberts}.}
  \bibinfo{year}{1956}\natexlab{}.
\newblock \showarticletitle{Note on Linear Forms}.
\newblock \bibinfo{journal}{\emph{Proc. Amer. Math. Soc.}}
  (\bibinfo{year}{1956}).
\newblock
\urldef\tempurl%
\url{https://doi.org/10.2307/2032755}
\showDOI{\tempurl}


\bibitem[Schroeder and Gibson(2010)]%
        {2010studyfailurehpc}
\bibfield{author}{\bibinfo{person}{Bianca Schroeder} {and}
  \bibinfo{person}{Garth~A. Gibson}.} \bibinfo{year}{2010}\natexlab{}.
\newblock \showarticletitle{A Large-Scale Study of Failures in High-Performance
  Computing Systems}.
\newblock \bibinfo{journal}{\emph{Transactions on Dependable and Secure
  Computing (TDSC)}} (\bibinfo{year}{2010}).
\newblock
\urldef\tempurl%
\url{https://doi.org/10.1109/TDSC.2009.4}
\showDOI{\tempurl}


\bibitem[Smith et~al\mbox{.}(2022)]%
        {2022turingnlg}
\bibfield{author}{\bibinfo{person}{Shaden Smith}, \bibinfo{person}{Mostofa
  Patwary}, \bibinfo{person}{Brandon Norick}, \bibinfo{person}{Patrick
  LeGresley}, \bibinfo{person}{Samyam Rajbhandari}, \bibinfo{person}{Jared
  Casper}, \bibinfo{person}{Zhun Liu}, \bibinfo{person}{Shrimai Prabhumoye},
  \bibinfo{person}{George Zerveas}, \bibinfo{person}{Vijay Korthikanti},
  \bibinfo{person}{Elton Zhang}, \bibinfo{person}{Rewon Child},
  \bibinfo{person}{Reza~Yazdani Aminabadi}, \bibinfo{person}{Julie Bernauer},
  \bibinfo{person}{Xia Song}, \bibinfo{person}{Mohammad Shoeybi},
  \bibinfo{person}{Yuxiong He}, \bibinfo{person}{Michael Houston},
  \bibinfo{person}{Saurabh Tiwary}, {and} \bibinfo{person}{Bryan Catanzaro}.}
  \bibinfo{year}{2022}\natexlab{}.
\newblock \bibinfo{title}{Using DeepSpeed and Megatron to Train Megatron-Turing
  NLG 530B, A Large-Scale Generative Language Model}.
\newblock
\newblock
\showeprint[arxiv]{2201.11990}~[cs.CL]


\bibitem[System(2023)]%
        {2023bamboogithub}
\bibfield{author}{\bibinfo{person}{UCLA System}.}
  \bibinfo{year}{2023}\natexlab{}.
\newblock \bibinfo{title}{Bamboo}.
\newblock
\newblock
\urldef\tempurl%
\url{https://github.com/uclasystem/bamboo}
\showURL{%
\tempurl}


\bibitem[Tarnawski et~al\mbox{.}(2021)]%
        {2021piper}
\bibfield{author}{\bibinfo{person}{Jakub~M Tarnawski}, \bibinfo{person}{Deepak
  Narayanan}, {and} \bibinfo{person}{Amar Phanishayee}.}
  \bibinfo{year}{2021}\natexlab{}.
\newblock \showarticletitle{Piper: Multidimensional Planner for DNN
  Parallelization}. In \bibinfo{booktitle}{\emph{Advances in Neural Information
  Processing Systems (NeurIPS)}}.
\newblock
\urldef\tempurl%
\url{https://proceedings.neurips.cc/paper_files/paper/2021/hash/d01eeca8b24321cd2fe89dd85b9beb51-Abstract.html}
\showURL{%
\tempurl}


\bibitem[Thorpe et~al\mbox{.}(2023)]%
        {2023bamboo}
\bibfield{author}{\bibinfo{person}{John Thorpe}, \bibinfo{person}{Pengzhan
  Zhao}, \bibinfo{person}{Jonathan Eyolfson}, \bibinfo{person}{Yifan Qiao},
  \bibinfo{person}{Zhihao Jia}, \bibinfo{person}{Minjia Zhang},
  \bibinfo{person}{Ravi Netravali}, {and} \bibinfo{person}{Guoqing~Harry Xu}.}
  \bibinfo{year}{2023}\natexlab{}.
\newblock \showarticletitle{Bamboo: Making Preemptible Instances Resilient for
  Affordable Training of Large DNNs}. In \bibinfo{booktitle}{\emph{USENIX
  Symposium on Networked Systems Design and Implementation (NSDI)}}.
\newblock
\urldef\tempurl%
\url{https://www.usenix.org/conference/nsdi23/presentation/thorpe}
\showURL{%
\tempurl}


\bibitem[Villalobos et~al\mbox{.}(2022)]%
        {2022modelsize}
\bibfield{author}{\bibinfo{person}{Pablo Villalobos}, \bibinfo{person}{Jaime
  Sevilla}, \bibinfo{person}{Tamay Besiroglu}, \bibinfo{person}{Lennart Heim},
  \bibinfo{person}{Anson Ho}, {and} \bibinfo{person}{Marius Hobbhahn}.}
  \bibinfo{year}{2022}\natexlab{}.
\newblock \bibinfo{title}{Machine Learning Model Sizes and the Parameter Gap}.
\newblock
\newblock
\showeprint[arxiv]{2207.02852}~[cs.LG]


\bibitem[Weng et~al\mbox{.}(2022)]%
        {2022mlaaswild}
\bibfield{author}{\bibinfo{person}{Qizhen Weng}, \bibinfo{person}{Wencong
  Xiao}, \bibinfo{person}{Yinghao Yu}, \bibinfo{person}{Wei Wang},
  \bibinfo{person}{Cheng Wang}, \bibinfo{person}{Jian He},
  \bibinfo{person}{Yong Li}, \bibinfo{person}{Liping Zhang},
  \bibinfo{person}{Wei Lin}, {and} \bibinfo{person}{Yu Ding}.}
  \bibinfo{year}{2022}\natexlab{}.
\newblock \showarticletitle{MLaaS in the Wild: Workload Analysis and Scheduling
  in Large-Scale Heterogeneous GPU Clusters}. In
  \bibinfo{booktitle}{\emph{USENIX Symposium on Networked Systems Design and
  Implementation (NSDI)}}.
\newblock
\urldef\tempurl%
\url{https://www.usenix.org/conference/nsdi22/presentation/weng}
\showURL{%
\tempurl}


\bibitem[Wolf et~al\mbox{.}(2020)]%
        {2022hftransformer}
\bibfield{author}{\bibinfo{person}{Thomas Wolf}, \bibinfo{person}{Lysandre
  Debut}, \bibinfo{person}{Victor Sanh}, \bibinfo{person}{Julien Chaumond},
  \bibinfo{person}{Clement Delangue}, \bibinfo{person}{Anthony Moi},
  \bibinfo{person}{Pierric Cistac}, \bibinfo{person}{Tim Rault},
  \bibinfo{person}{Remi Louf}, \bibinfo{person}{Morgan Funtowicz},
  \bibinfo{person}{Joe Davison}, \bibinfo{person}{Sam Shleifer},
  \bibinfo{person}{Patrick von Platen}, \bibinfo{person}{Clara Ma},
  \bibinfo{person}{Yacine Jernite}, \bibinfo{person}{Julien Plu},
  \bibinfo{person}{Canwen Xu}, \bibinfo{person}{Teven Le~Scao},
  \bibinfo{person}{Sylvain Gugger}, \bibinfo{person}{Mariama Drame},
  \bibinfo{person}{Quentin Lhoest}, {and} \bibinfo{person}{Alexander Rush}.}
  \bibinfo{year}{2020}\natexlab{}.
\newblock \showarticletitle{Transformers: State-of-the-Art Natural Language
  Processing}. In \bibinfo{booktitle}{\emph{Conference on Empirical Methods in
  Natural Language Processing (EMNLP)}}.
\newblock
\urldef\tempurl%
\url{https://doi.org/10.18653/v1/2020.emnlp-demos.6}
\showDOI{\tempurl}


\bibitem[Xie et~al\mbox{.}(2020)]%
        {2020elan}
\bibfield{author}{\bibinfo{person}{Lei Xie}, \bibinfo{person}{Jidong Zhai},
  \bibinfo{person}{Baodong Wu}, \bibinfo{person}{Yuanbo Wang},
  \bibinfo{person}{Xingcheng Zhang}, \bibinfo{person}{Peng Sun}, {and}
  \bibinfo{person}{Shengen Yan}.} \bibinfo{year}{2020}\natexlab{}.
\newblock \showarticletitle{Elan: Towards Generic and Efficient Elastic
  Training for Deep Learning}. In \bibinfo{booktitle}{\emph{International
  Conference on Distributed Computing Systems (ICDCS)}}.
\newblock
\urldef\tempurl%
\url{https://doi.org/10.1109/ICDCS47774.2020.00018}
\showDOI{\tempurl}


\bibitem[Zhang et~al\mbox{.}(2022)]%
        {2022opt}
\bibfield{author}{\bibinfo{person}{Susan Zhang}, \bibinfo{person}{Stephen
  Roller}, \bibinfo{person}{Naman Goyal}, \bibinfo{person}{Mikel Artetxe},
  \bibinfo{person}{Moya Chen}, \bibinfo{person}{Shuohui Chen},
  \bibinfo{person}{Christopher Dewan}, \bibinfo{person}{Mona Diab},
  \bibinfo{person}{Xian Li}, \bibinfo{person}{Xi~Victoria Lin},
  \bibinfo{person}{Todor Mihaylov}, \bibinfo{person}{Myle Ott},
  \bibinfo{person}{Sam Shleifer}, \bibinfo{person}{Kurt Shuster},
  \bibinfo{person}{Daniel Simig}, \bibinfo{person}{Punit~Singh Koura},
  \bibinfo{person}{Anjali Sridhar}, \bibinfo{person}{Tianlu Wang}, {and}
  \bibinfo{person}{Luke Zettlemoyer}.} \bibinfo{year}{2022}\natexlab{}.
\newblock \bibinfo{title}{OPT: Open Pre-trained Transformer Language Models}.
\newblock
\newblock
\showeprint[arxiv]{2205.01068}~[cs.CL]


\bibitem[Zhang et~al\mbox{.}(2023)]%
        {2023ecrec}
\bibfield{author}{\bibinfo{person}{Tianyu Zhang}, \bibinfo{person}{Kaige Liu},
  \bibinfo{person}{Jack Kosaian}, \bibinfo{person}{Juncheng Yang}, {and}
  \bibinfo{person}{Rashmi Vinayak}.} \bibinfo{year}{2023}\natexlab{}.
\newblock \showarticletitle{Efficient Fault Tolerance for Recommendation Model
  Training via Erasure Coding}.
\newblock \bibinfo{journal}{\emph{Proceedings of the VLDB Endowment (VLDB)}}
  (\bibinfo{year}{2023}).
\newblock
\urldef\tempurl%
\url{https://doi.org/10.14778/3611479.3611514}
\showDOI{\tempurl}


\bibitem[Zhao et~al\mbox{.}(2023)]%
        {2023fsdp}
\bibfield{author}{\bibinfo{person}{Yanli Zhao}, \bibinfo{person}{Andrew Gu},
  \bibinfo{person}{Rohan Varma}, \bibinfo{person}{Liang Luo},
  \bibinfo{person}{Chien-Chin Huang}, \bibinfo{person}{Min Xu},
  \bibinfo{person}{Less Wright}, \bibinfo{person}{Hamid Shojanazeri},
  \bibinfo{person}{Myle Ott}, \bibinfo{person}{Sam Shleifer},
  \bibinfo{person}{Alban Desmaison}, \bibinfo{person}{Can Balioglu},
  \bibinfo{person}{Bernard Nguyen}, \bibinfo{person}{Geeta Chauhan},
  \bibinfo{person}{Yuchen Hao}, {and} \bibinfo{person}{Shen Li}.}
  \bibinfo{year}{2023}\natexlab{}.
\newblock \bibinfo{title}{PyTorch FSDP: Experiences on Scaling Fully Sharded
  Data Parallel}.
\newblock
\newblock
\showeprint[arxiv]{2304.11277}~[cs.DC]


\bibitem[Zheng et~al\mbox{.}(2022)]%
        {2022alpa}
\bibfield{author}{\bibinfo{person}{Lianmin Zheng}, \bibinfo{person}{Zhuohan
  Li}, \bibinfo{person}{Hao Zhang}, \bibinfo{person}{Yonghao Zhuang},
  \bibinfo{person}{Zhifeng Chen}, \bibinfo{person}{Yanping Huang},
  \bibinfo{person}{Yida Wang}, \bibinfo{person}{Yuanzhong Xu},
  \bibinfo{person}{Danyang Zhuo}, \bibinfo{person}{Eric~P. Xing},
  \bibinfo{person}{Joseph~E. Gonzalez}, {and} \bibinfo{person}{Ion Stoica}.}
  \bibinfo{year}{2022}\natexlab{}.
\newblock \showarticletitle{Alpa: Automating Inter- and {Intra-Operator}
  Parallelism for Distributed Deep Learning}. In
  \bibinfo{booktitle}{\emph{USENIX Symposium on Operating Systems Design and
  Implementation (OSDI)}}.
\newblock
\urldef\tempurl%
\url{https://www.usenix.org/conference/osdi22/presentation/zheng-lianmin}
\showURL{%
\tempurl}


\end{thebibliography}
}
\clearpage

\section*{Appendix}
\appendix
\section{Proof of Nodes Specification Covering All Nodes}
\label{sec:apdx_proof_frobenius}

We prove the following theorem which shows a finite set of $p$ number of pipeline templates, where the number of nodes is $(n_0, n_1, \dots, n_{p-1})$ ($n_i < n_{i+1}$), can fully cover the node cluster with feasible number of nodes $N'$ any time irrespective of how many failures happen on the cluster as long as its feasibility holds.

\begin{theorem}
    $N'$ nodes ($(f+1)n_0 \le N' \le N$) can always be represented as a linear combination of the $p$ pipeline templates with $(n_0, n_1, \dots, n_{p-1})$ number of nodes, respectively,
    if the following two conditions are satisfied:
    \begin{denseenum}
        \item $p > n_0 - 1$. \label{eq:apdx_req1}
        \item $n_i$ are consecutive integers ($n_i + 1 = n_{i+1}$) \label{eq:apdx_req2}.
    \end{denseenum}
\end{theorem}

\begin{proof}
    We first formulate an integer linear combination to represent $N'$:

    \begin{equation}
        N' = x_0n_0 + x_1n_1 + \dots, x_{p-1}n_{p-1}
        \label{eq:apdx_linearcombination}
    \end{equation}
    where $x_i$ is the number of pipelines to be instantiated from the pipeline template with $n_i$ number of nodes.

    The Frobenius number $g(n_0, n_1, \dots, n_{p-1})$, the largest number that cannot be represented as a linear combination of Equation~\ref{eq:apdx_linearcombination}, has proven to be:
    \begin{equation}
        g = \left( \left\lfloor \frac{n_0 - 2}{p-1} \right\rfloor \right) + d(n_0 - 1)
    \end{equation}
    if the integer set $(n_0, n_1, \dots, n_{p-1})$ is an arithmetic sequence, \ie, $n_i = n_0 + d(i-1)$~\cite{2005frobenius}.

    When we apply both Requirements \ref{eq:apdx_req1} and \ref{eq:apdx_req2}, $g=n_0-1$.
    Fault tolerance threshold $f$ is a non-negative integer; the minimum feasible number of nodes $N'=(f+1)n_0$ is $n_0$.
    Therefore, any feasible $N'$ that is larger than $g$ and can be represented as a linear combination of Equation~\ref{eq:apdx_linearcombination}.
\end{proof}



\section{Proof of Guarantee for Pipeline Template Availability When Merging Pipelines}
\label{sec:apdx_pipeline_merge}
We first show this when failures happen in a single pipeline.
When we lose $k$ nodes ($k > 0$) from a pipeline, where all pipelines have $n_0$ nodes and are not able to yield any node, \name instantiates a new pipeline with $2n_0-k$ nodes by merging it with another $n_0$-node pipeline.
We prove that a pipeline template with $2n_0-k$ nodes is always available.
\begin{theorem}
    A set of pipeline templates always includes a pipeline template with $2n_0-k$ nodes ($2n_0-k \ge n_0$).
\label{theorem:pipeline_merge}
\end{theorem}
\begin{proof}
    A set of pipeline templates has pipeline templates with up to $N-fn_0$ nodes (\S\ref{sec:node_specification}).
    Assume that we do not have a pipeline template with $2n_0-k$ number of nodes specification, then $N-fn_0 < 2n_0-k$ and $N < (f+2)n_0-k$ are assumed to be true.
    To not break the fault tolerance threshold that we maintain at least $f+1$ model replicas after merging two pipelines, we must have at least $f+2$ replicas, \ie, we should have had at least $(f+2)n_0$ nodes before failures.
    Since the initial number of nodes $N$ is always larger than the number of currently remaining nodes, $N > (f+2)n_0$ inequality holds and it contradicts our initial assumption.
    Therefore, we have a pipeline template with $2n_0-k$ number of nodes.
\end{proof}

When failures happen across several pipelines, multiple pipelines can have less than $n_0$ nodes.
\name repeatedly merges two pipelines until a new pipeline has enough number of nodes.
Assume we merged $m$ pipelines to get enough number of nodes, i.e., $\sum_{i=0}^{m}{n_{p_i}} \ge n_0$.
It means merging $m-1$ pipelines was not enough to get $n_0$ nodes, i.e., $\sum_{i=0}^{m-1}{n_{p_i}} < n_0$.
With $n_{p_m} \le n_0$, we have an inequality $n_0 \le \sum_{i=0}^{m-1}{n_{p_i}} + n_{p_m} < 2n_0$.
It has already been proved by Theorem~\ref{theorem:pipeline_merge} that we have a pipeline template for all numbers in the range.

\section{Throughput of All Models in Spot Instances}
\label{sec:apdx_spot_instances}

\begin{figure}
    \centering
    \begin{subfigure}[t]{0.8\linewidth}\
        \centering
        \includegraphics[width=\linewidth]{evaluations/gcp_res/legend.pdf}
    \end{subfigure}
    \begin{subfigure}[t]{0.5\linewidth}
        \includegraphics[width=\linewidth]{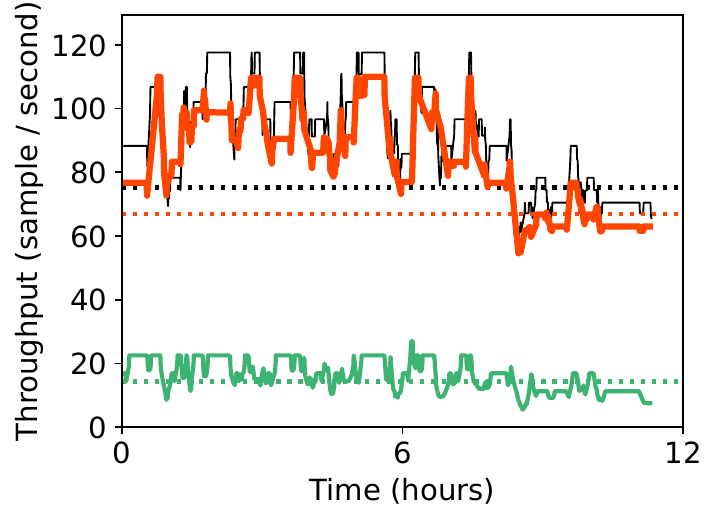}
        \caption{GPT-2 in \texttt{P3}.}
    \end{subfigure}%
    \begin{subfigure}[t]{0.5\linewidth}
        \includegraphics[width=\linewidth]{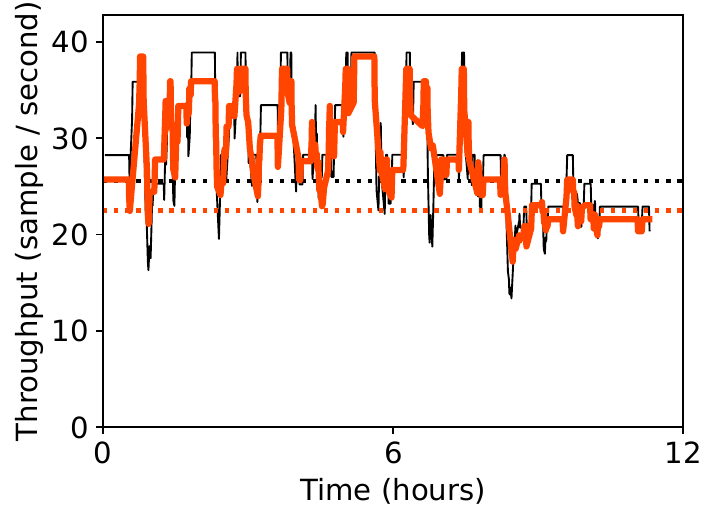}
        \caption{GPT-3 Medium in \texttt{P3}.}
    \end{subfigure}

    \begin{subfigure}[t]{0.5\linewidth}
        \includegraphics[width=\linewidth]{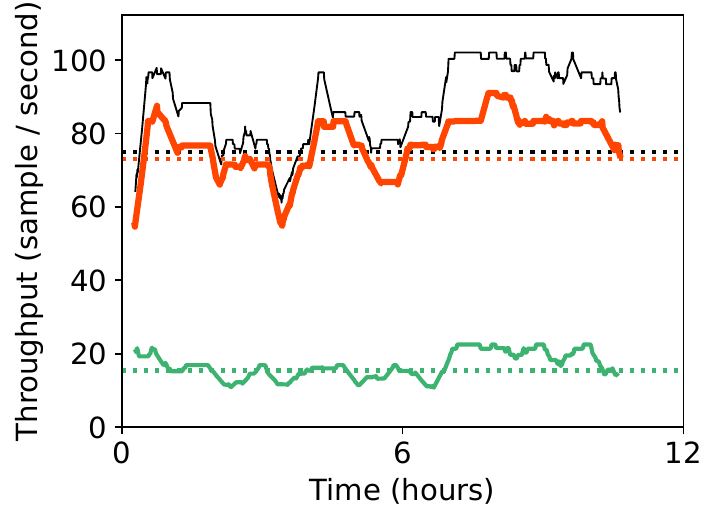}
        \caption{GPT-2 in \texttt{a2-highgpu-1g}.}
    \end{subfigure}%
    \begin{subfigure}[t]{0.5\linewidth}
        \includegraphics[width=\linewidth]{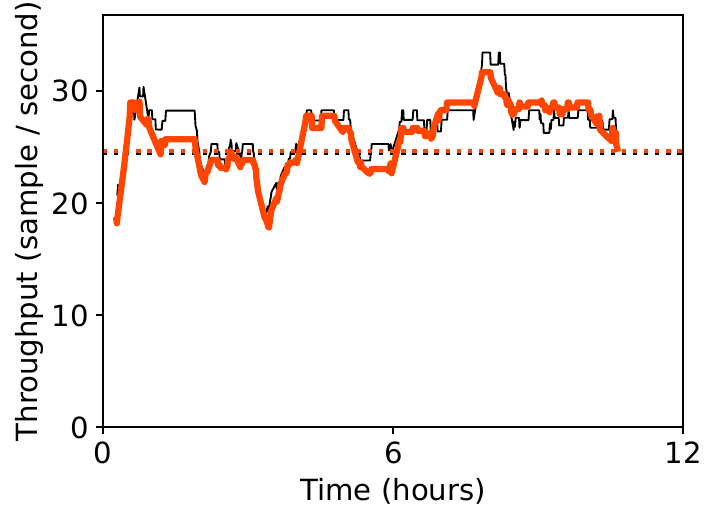}
        \caption{GPT-3 Medium in \texttt{a2-highgpu-1g}.}
    \end{subfigure}

    \caption{GPT-2 and GPT-3 medium throughput changes in Amazon EC2 \texttt{P3} (top) and Google \texttt{a2-highgpu-1g} (bottom) instances.}
    \label{fig:apdx_throughput}
\end{figure}

Figure~\ref{fig:apdx_throughput} shows throughput of unpresented models in the paper due to lack of space, running on Amazon EC2 \texttt{P3} spot instances and Google \texttt{a2-highgpu-1g} spot instances.
Varuna could avoid fallback overhead by successfully checkpointing ahead of preemption in small models, (\eg, BERT-large, GPT-2, and GPT-3 medium), thus Varuna throughput matches \name on average.
However, in large models, \name outperforms it.
Note that lines are smoothed for visibility and do not precisely represent throughput.
Varuna has more spikes down to 0 throughput in reality thus has less throughput.

\end{document}